\documentclass[lettersize,journal]{IEEEtran}
\usepackage[utf8]{inputenc}
\usepackage{indentfirst}
\usepackage{amsfonts}
\usepackage{amssymb}
\usepackage{algorithm}
\usepackage{algorithmic}
\usepackage{graphicx}
\usepackage{float} 
\usepackage{subfigure}
\usepackage{caption}
\usepackage{cite}
\usepackage{verbatim}
\usepackage{diagbox}
\usepackage{threeparttable}
\usepackage{multirow}
\usepackage{stfloats}
\usepackage{makecell}
\usepackage{amsmath}
\usepackage{amsthm}
\newtheorem{theorem}{Theorem}
\newtheorem{corollary}[theorem]{\indent Corollary}

\newtheorem{proposition}[theorem]{\indent Proposition}

\newtheorem{remark}{Remark}

\ifCLASSINFOpdf

\fi

\begin{document}

    \title{Reconfigurable Distributed Antennas and Reflecting Surface: A New Architecture for Wireless Communications }
% ~\IEEEmembership{Member,~IEEE}
\author{Chengzhi Ma,
        Xi Yang,
        Jintao Wang,
        Guanghua Yang,
        Wei Zhang,
        and Shaodan Ma
        % <-this % stops a space
\thanks{C. Ma, J. Wang and S. Ma are with the State Key Laboratory of Internet of Things for Smart City and the Department
of Electrical and Computer Engineering, University of Macau, Macao SAR, China (e-mails: yc07499@um.edu.mo; 
wang.jintao@connect.um.edu.mo;
shaodanma@um.edu.mo).}% <-this % stops a space
\thanks{X. Yang is with the Shanghai Key Laboratory of Multidimensional Information Processing, School of Communication and Electronic Engineering, East China Normal University, Shanghai 200241, China (e-mail: xyang@cee.ecnu.edu.cn).}
\thanks{G. Yang is with the School of Intelligent Systems Science and Engineering and GBA and B{\&}R International Joint Research Center for Smart Logistics, Jinan University, Zhuhai 519070, China (e-mail: ghyang@jnu.edu.cn).}
\thanks{W. Zhang is with the School of Electrical Engineering \& Telecommunications, University of New South Wales, Sydney, Australia (e-mail: w.zhang@unsw.edu.au).}
}
\UseRawInputEncoding

% The paper headers
% \markboth{Journal of \LaTeX\ Class Files,~Vol.~14, No.~8, August~2015}%
% {Shell \MakeLowercase{\textit{et al.}}: Bare Demo of IEEEtran.cls for IEEE Journals}

\maketitle
% As a general rule, do not put math, special symbols or citations
% in the abstract or keywords.
\begin{abstract}

% Distributed antenna systems (DAS) can provide uniform coverage with the potential to reduce power and utilize multiple spatially distributed remote radio units (RRUs). On the other hand, reconfigurable intelligent surfaces (RISs) have the potential to improve the spectral and energy efficiencies of wireless communication by leveraging a large number of reconfigurable units for passive beamforming gain. However, various challenges are expected when directly adding RIS into the current DAS in practical deployment, e.g., a large number of elements on RIS are needed to combat the ``multiplicative fading" effect, and high-quality control link is required for configuration to fully unleash the passive beamforming gain of RIS. 

Distributed Antenna Systems (DASs) employ multiple antenna arrays in remote radio units to achieve highly directional transmission and provide great coverage performance for future-generation networks. However, the utilization of fully digital or hybrid active antenna arrays results in a significant increase in hardware costs and power consumption for DAS. To address these issues, integrating DAS with Reconfigurable Intelligent Surfaces (RIS) offers a viable approach to ensure coverage and transmission performance while maintaining low hardware costs and power consumption. To incorporate the merits of RIS into the DAS from practical consideration, a novel architecture of ``Reconfigurable Distributed Antennas and Reflecting Surfaces (RDARS)'' is proposed in this paper. Specifically, based on the design of the additional direct-through state together with the existing high-quality fronthaul link, any element of the RDARS can be dynamically programmed to connect with the base station (BS) via fibers and perform the \textit{connected mode} as remote distributed antennas of the BS to receive or transmit signals. Additionally, RDARS also inherits the low-cost and low-energy-consumption benefits of fully passive RISs by default configuring the elements as passive to perform the \textit{reflection mode}. As a result, RDARS encompasses both DAS and RIS as special cases, offering flexible control over the trade-off between \textit{distribution gain} and \textit{reflection gain} to enhance performance. To unveil the potential of such architecture, the ergodic achievable rate under the RDARS architecture is analyzed and closed-form expression with meaningful insights is derived. The theoretical analysis proves that the RDARS can achieve a higher achievable rate than both DAS and fully passive RIS with the passive beamforming gain provided by elements acting \textit{reflection mode} while combating the ``multiplicative fading'' suffered by RISs through the \textit{connected mode} performed at the RDARS. Simulation results also demonstrate the superiority of the RDARS architecture over DAS and passive RIS-aided systems and its flexible trade-off between performance and cost. To further validate the feasibility and effectiveness, an RDARS prototype with 256 elements is built for real experiments. Experimental results show that the RDARS-aided system with only one element operating in connected mode can achieve an additional 21\% and 170\% throughput improvement over DAS and RIS-aided systems, respectively.

\end{abstract}


\IEEEpeerreviewmaketitle

\section{Introduction}

    Distributed antenna system (DAS) is one of the promising wireless communication architectures where multiple remote radio units (RRUs) distributed spatially throughout the network are connected to a home base station (BS) for signal processing\cite{Robert_Heath}.\footnote{The deployment of DAS depends on network structure, and can be implemented in various methods, e.g., the RRUs can also be connected to a central processing units (CPU) \cite{Lisu_Yu, Lin_Dai}. As such, we take one of the definitions of DAS as an example here\cite{Robert_Heath}. } By leveraging state-of-art radio over fiber (ROF) technologies \cite{ROF1, ROF2, Lisu_Yu}, the DAS can potentially reduce the average distance of propagation to or from the nearest antenna, thus reducing the required power for transmission by creating uniform coverage with micro/macro-diversity gain as compared to co-located antenna deployment \cite{Feng_Wei}. Though the original idea of DAS was to improve indoor cellular communication by increasing coverage and reducing the outage without introducing extensive BS deployments \cite{Saleh}, prior studies have also identified its potential to improve system energy efficiency and system capacity \cite{Robert_Heath, Lin_Dai, Lin_Dai2}. For example, a code division multiple access (CDMA)-DAS with maximal ratio combining (MRC) was analyzed in \cite{Lin_Dai}, where the benefits of macro-diversity introduced by DAS were revealed. The superiority of DAS as compared to co-located BS antennas was demonstrated theoretically in \cite{Lin_Dai2}. As a result, the DAS architecture has been applied in various systems and scenarios\cite{6G}, e.g., the cloud radio access network and ultra-dense network. Moreover, by utilizing the inherent macro-diversity provided by DAS with the merits introduced by massive multiple-input-multiple-output (MIMO) \cite{CZ}, the combination of DAS and linear signal processing technologies has given rise to a new communication paradigm, i.e., cell-free massive MIMO (CF-mMIMO)\cite{Ngo}. CF-mMIMO can alleviate interference and provide uniformly good service and thus is anticipated to play an essential role in future wireless communication systems. 
    
    % CF-mMIMO can alleviate interference and provide uniformly good service and thus is anticipated to play an essential role in future wireless communication systems. Nevertheless, deploying massive MIMO in each RRU will bring tremendous hardware costs and energy consumption.

    Despite the significant advantages exemplified above, many densely deployed antennas or remote radio units (RRUs) are needed, and large antenna arrays are preferred to be deployed in these distributed units to obtain sufficient antenna gain and alleviate interference\cite{Lisu_Yu}. This undoubtedly introduces tremendous hardware costs and energy consumption, thus limiting its practical applications. On the other hand, reconfigurable intelligent surface (RIS), also known as intelligent reflecting surface (IRS), has been proposed as a promising technology to extend the coverage while maintaining a low cost and energy consumption for improving spectrum efficiency (SE) and energy efficiency (EE) of the communication network \cite{Qingqing_Wu, Chongwen_Huang, LIS}. Specifically, RIS is a planar array composed of many passive elements, each of which can impose an independent phase shift on the incident signals. By smartly tuning the phase shifts, the wireless propagation environment becomes controllable and programmable, which thus provides an extra degree of freedom to optimize communication performance. 

    However, applying RIS for current wireless communication systems faces great challenges due to RIS's inherent characteristics. For example, the RIS reflective link generally suffers from the ``multiplicative fading'' effect, i.e., the equivalent channel gain of the transmitter-RIS-receiver link is the product of the transmitter-RIS and RIS-receiver channel gains \cite{Zijian_Zhang}. Therefore, the path loss over the RIS reflective link may be high due to the ``multiplicative fading'', resulting in an overall small performance gain. The merits of directly adding RIS into current wireless communication systems may be only exhibited when the direct link is relatively weak or suffers from blockage\cite{Van}. As a result, many RIS elements are generally needed for performance enhancement. However, the high signaling overhead introduced by the training pilots required for channel estimation and channel information feedback \cite{Hu_Chen} as well as the high complexity in real-time beamforming impose great challenges for RIS deployment with a large number of passive elements in practical wireless network \cite{Cunhua_Pan}. Moreover, from the implementation perspective, high-quality feedback or control link for delivering the control signal for phase shift design is required to fully leverage the passive beamforming gain of RIS.
    Although some RIS variants, e.g., active RIS \cite{Zijian_Zhang} or hybrid relay-reflecting intelligent surface (HR-RIS) \cite{Nguyen} were proposed to mitigate some of these limitations. Nevertheless, these variants require additional reflection-type amplifiers or relaying to replace the original passive elements for performance enhancement, which is highly dependent on the availability of external power supply.\footnote{The signal model of RDARS-aided system is also different from active RIS\cite{Zijian_Zhang} or HR-RIS-aided system\cite{Nguyen}. This will be demonstrated by \eqref{system_model_SISO} in Section III.}
    
    As such, motivated by the design philosophy of RIS in terms of ``reconfigurability" as well as its ability to provide passive beamforming gain with low costs\cite{Qingqing_Wu}, it is of great interest to rethink the architecture of DAS and incorporate the merits of RIS into DAS to fulfill the unprecedented requirements for future communication, e.g., ubiquitous coverage, ultra-high-rate throughput, and so on \cite{6G, Qing_Xue}, but with low hardware costs and energy consumption.
    These, thus, motivate us to think outside the box and to revisit the design philosophy of both systems, i.e., instead of simply adding RIS into DAS, \textit{How about we empower RRUs in DAS with the ability for passive beamforming by low-cost programmable reflection elements while maintaining the distribution gain via programmable selected active antenna elements?}
    % Therefore both industry and academia have started to look for new disruptive technologies to combine with existing DAS to fulfill the unprecedented requirements for future communication aiming at ubiquitous coverage, ultra-high-rate throughput,   and so on\cite{6G, Qing_Xue}.
        
    % For example, the study in \cite{Kangda_Zhi} shows that by combing RIS with a massive MIMO system, the performance can be greatly enhanced by applying a two-time scale design, i.e., the base station (BS) is adapted to the instantaneous CSI while the passive beamforming performing at RIS is configured according to statistical CSI. Also, it is shown that even with phase noise and imperfect CSI, through proper joint optimization and efficient design algorithm, the total average mean-square-error (MSE) can be greatly reduced in the RIS-aided system, and incorporating RIS is still capable of boosting system performance in such system setup\cite{Jintao_Wang}.

    % Besides, conventional RIS also faces challenges in channel estimations since the number of channel coefficients to be estimated scales linearly with the product of the numbers of antennas at the BS, UE, and the elements on RIS \cite{YinHaifan}. Even with appropriate design of the channel estimation schemes and leveraging the channel characteristics, e.g., sparsity, the pilot overhead is still large and becomes unaffordable when deploying many passive elements for performance enhancement\cite{Wang_Peilan}. 
    
    To this end, Reconfigurable Distributed Antennas and Reflecting Surface (RDARS) is proposed in this paper as a promising new architecture that exhibits the merits of both DAS and RIS while addressing their limitations from a practical implementation perspective. Specifically, a RDARS is a planar array consisting of reconfigurable elements, where each element can switch between two modes, namely, the \textit{reflection mode} and the \textit{connected mode}. When working under \textit{reflection mode}, the elements work as passive reflecting elements similar to the elements in fully passive RIS. When working under \textit{connected mode}, the elements act as distributed antennas connected to the BS via dedicated wires or fibers and can receive the incoming or transmit wireless signal. As a result, an element working under \textit{connected mode} on RDARS can be viewed as a remote antenna in traditional DAS\cite{Robert_Heath, Lisu_Yu}. Consequently, different from conventional DAS, where all the antenna elements are connected with energy-demanding components, e.g., radio frequency (RF) chains, some of them are envisioned to be replaced with low-cost passive reflection elements to reduce the hardware complexity and cost while still maintaining satisfying performance by leveraging passive beamforming gain. Moreover, leveraging the ``reconfigurability'' introduced by RDARS with dynamically configured different modes for each element, RDARS is envisioned to provide more degree of freedom and flexibility for system design with higher performance gain than its counterparts. 

    \subsection{Main Contributions}
    To unveil the performance gain introduced by RDARS, this paper focuses on proposing the concept and unveiling its potential with \textit{theoretical analysis} and then presents a showcase of the RDARS prototype with \textit{experimental results}. The main contributions of the paper are as follows:
        \begin{itemize}
            \item We propose a novel RDARS architecture for future wireless communications. One key feature of RDARS is the reconfigurability, i.e., each element of RDARS can be dynamically programmed to perform either two modes, i.e., 
            perform as a remote antenna to receive or transmit signals under \textit{connected mode} or reflect signals as a conventional fully-passive RIS element under \textit{reflection mode}. As a result, RDARS serves as a flexible and reconfigurable combination of distributed antennas and reflecting surfaces. It covers the DAS and RIS as special cases with great potential for performance enhancement and practical deployment. 

            \item We first analyze the performance of the RDARS-aided system under a single-antenna BS scenario with optimal phase shift design for RDARS configuration based on instantaneous channel state information (CSI). Closed-form expression for ergodic achievable rate is derived under the Rayleigh fading channel. The analytical results indicate that the RDARS-aided system always outperforms the DAS counterpart and RIS-aided systems with practical surface size.
            
            \item We further characterize the performance of the RDARS-aided system considering multi-antenna BS with arbitrary RDARS configuration. Closed-form approximations of ergodic achievable rate are provided where the more general Rician fading channel with line-of-sight (LoS) component between UE-RDARS and RDARS-BS links are considered.

            \item We provide extensive simulation results to validate the analysis and investigate the impact of various parameters with clear insights. The results confirm the positive benefits of the RDARS-aided system, e.g., the transmit power and the required number of antennas/RF chains can be reduced significantly compared to the counterparts. 
            
            % in certain scenarios, e.g., the required total number of antennas at BS and elements on the surface can be reduced with an RDARS-aided system while guaranteeing certain performance gain. 
            
            \item We conduct practical experiments to demonstrate one showcase of RDARS's applications by establishing a RDARS-aided communication system prototype with $256$ configurable elements. The experimental results show the superiority of the proposed RDARS in terms of improvements in the throughput, the received SNR, and the block error rate (BLER), thus serving as vivid proof of performance enhancement and manufacturability.   
        \end{itemize}
        
    The rest of the paper is organized as follows. In Sections II and III, we propose the idea of the RDARS architecture and the system model of the RDARS-aided system with single-antenna BS. Then, the performance of the RDARS-aided system is analyzed in Section IV. In Section V, we extend the performance analysis to multi-antenna BS scenarios. Simulation results are given in Section VI, followed by experimental results in Section VII. The conclusion is drawn in Section VIII.
    
    \textit{Notations:} Throughout the paper, numbers, vectors, and matrices are denoted by lower-case, bold-face lower-case, and bold-face upper-case letters, respectively. $(\cdot)^{T}$, $(\cdot)^{H}$ denote the transpose and conjugate transpose of a matrix or vector. ${\rm{diag}}(\mathbf{v})$ denotes a diagonal matrix with each diagonal element being the corresponding element in $\mathbf{v}$. Similarly, ${\rm{blk}}(\cdot)$ denote the block-diagonalization operation. Furthermore, $|\cdot|$ denotes the modulus of a complex number. ${\rm{Tr}}(\cdot)$ and $\mathbb{E}[\cdot]$ represents trace and expectation operator, respectively. The element in the $i^{th}$ row and $j^{th}$ column of matrix $\mathbf{A}$ is denoted as $\mathbf{A}(i,j)$. Finally, a Gamma-distributed random variable $x$ is denoted as $x \sim \Gamma(k,p)$, where $k$ and $p$ are the shape and scale parameters, respectively.

\section{RDARS}

    The proposed RDARS consists of a total of $N$ elements, and each element can switch between two modes, namely, the \textit{connected mode} and the \textit{reflection mode} as shown in Fig. 1. The RDARS controller can dynamically configure the specific mode of each element. When working under \textit{connected mode}, the elements act as a remote antenna connected to the BS via dedicated wires or fibers and can receive the incoming or transmit wireless signal.\footnote{In this paper, we focus on the characterizing the performance of RDARS in the uplink where the element acting as \textit{connected mode} is used for receiving the incoming wireless signal \cite{ROF3}. The extension to RDARS acting as one of the transmission points in the downlink will be our future work.} When working under \textit{reflection mode}, the elements work as passive reflecting units similar to the elements in fully passive RIS. As a result, RDARS serves as a flexible and reconfigurable combination of distributed antennas and reflecting surfaces. To reduce the hardware costs, we only consider the case where the number of elements performing \textit{connected mode} is very small, i.e., $a << N$. To better characterize the configuration of this architecture, we introduce an indicating matrix $\mathbf{A} \in \mathbb{C}^{a \times N}$ and $\mathbf{A}^{H}\mathbf{A}(i,i) \in \{0,1\}$, where $\mathbf{A}^{H}\mathbf{A}(i,i) = 1$ means the $i^{th}$ element works at \textit{connected mode}, while $\mathbf{A}^{H}\mathbf{A}(j,j) = 0$ indicates that the $j^{th}$ element operates at \textit{reflection mode}, i.e., the $j^{th}$ element can only alter the phase of incoming signals.
    Let $\mathbf{\Theta} = {\rm{diag}}(\boldsymbol{\theta}) \in \mathbb{C}^{N \times N}$ denote the RDARS reflection-coefficient matrix whose $i^{th}$ diagonal element is expressed as $\mathbf{\Theta}(i,i) = \boldsymbol{\theta}(i) =e^{j{\rm{arg}}(\boldsymbol{\theta}(i))}, \forall i$, and ${\rm{arg}}(\boldsymbol{\theta}(i))$ is the phase shift the $i^{th}$ element will induce.
    As such, the equivalent RDARS reflection-coefficient matrix can be modeled as $\mathbf{B} \triangleq (\mathbf{I}-\mathbf{A}^{H}\mathbf{A}) \mathbf{\Theta} \in \mathbb{C}^{N \times N}$. 

    An implementation example of RDARS is presented in Section VII. It should be noted that unlike the design philosophy of active RIS \cite{Zijian_Zhang} or HR-RIS \cite{Nguyen}, the key idea of RDARS is to build upon the DAS with low-cost passive elements for passive beamforming gain and to leverage the existing high-quality fronthaul link with ROF technologies\cite{Lisu_Yu} as exemplified in Fig. 2. Accordingly, such architecture demonstrates significant potential for practical deployment, i.e., the control link of RDARS can be integrated with the fronthaul of the current DAS leveraging the wildly applied ROF technologies\cite{Lisu_Yu, ROF1, ROF2}. This, thus guarantees the high-quality feedback and control signal for RDARS configuration to fully unleash the potential of passive beamforming gain.
    Besides, by leveraging the signal received by the elements acting \textit{connected mode}, RDARS has the potential to obtain ``global CSI'' with the method proposed in \cite{Abdelrahman, YinHaifan}.\footnote{
    In \cite{YinHaifan} and \cite{Abdelrahman}, ``semi-passive RIS'' was proposed merely for channel estimation, and the received pilot signals were mainly processed locally on RIS with low-cost signal processing units. This may limit the estimation accuracy. On the contrary, the RDARS-aided system has the freedom to jointly process signals at BS (or CPU) with more available resources and signal processing units. The study of RDARS for CSI acquisition will be our future research topic. 
    }
    The ``global CSI'', i.e.,  the individual acquisition of BS-RDARS and RDARS-User equipment(UE) channels, is infeasible for traditional RIS to acquire due to the inherent scaling ambiguity issue\cite{Qingqing_Wu, RIS_Overview}. 
    However, such ``global CSI'' can be obtained by RDARS and to further extract useful information about the surrounding environment. This, thus, enables a wide range of sensing applications. As a matter of fact, the RDARS can also be applied for integrated sensing and communication (ISAC) applications to facilitate reliable user localization without compromising the communication rate as in\cite{RDARS_demo}. On top of that, RDARS benefits from the added advantage of "reconfigurability," which allows for dynamically configured different modes for each element. This feature grants RDARS more freedom and flexibility in system design, making it highly anticipated to deliver significant performance enhancements compared to its counterparts with more potential applications.
    % \footnote{After the initial version of this paper was submitted on March 13, 2023, more experiments regarding RDARS-aided ISAC are conducted, details can be found on \cite{RDARS_demo}.}

    % Denote the number of elements working on \textit{connected mode} as $a$. Then, for $a = 0$, the RDARS-aided system becomes a fully passive RIS-aided system; for $a = N$, the RDARS-aided system can be interpreted as a DAS that has two distributed antenna sets, and one set has $N$ antennas located at a remote place while the other set of antennas is deployed at the BS \cite{Junyuan_Wang}. 

    \begin{figure} [htb]
    	\centering
            % \vspace{-0.8cm}
            % \setlength{\belowcaptionskip}{-0.5cm}
            \includegraphics[width=0.5\textwidth]{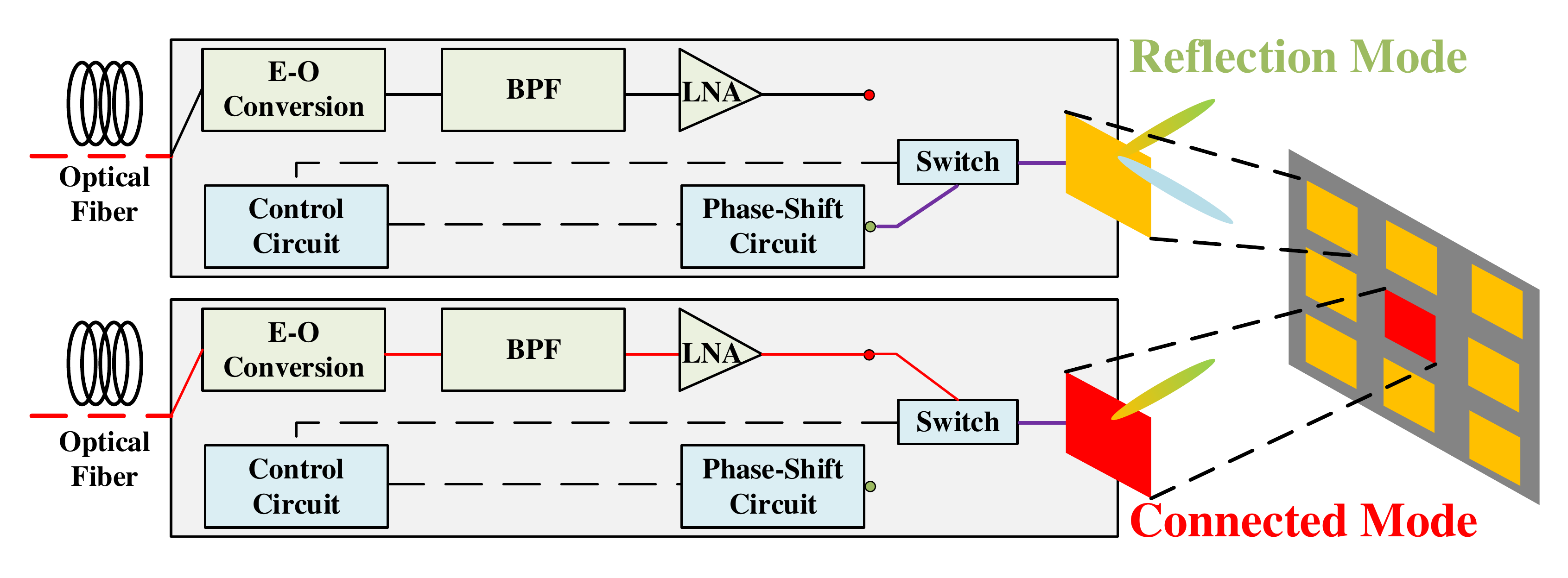}
        \caption{ An implementation example of RDARS. The \textit{connected mode} is built with the electrical-optical interface (E-O Conversion), the bandpass filter (BPF), and low-noise amplifier (LNA) by leveraging existing high-quality fronthaul links (fiber or cable) using ROF technologies\cite{Lisu_Yu}.}
    \end{figure}

    \begin{figure*} [htb]  
        \centering
        \setlength{\belowcaptionskip}{-3mm}
        \subfigure[uplink transmission] {
        \includegraphics[width=0.85\columnwidth]{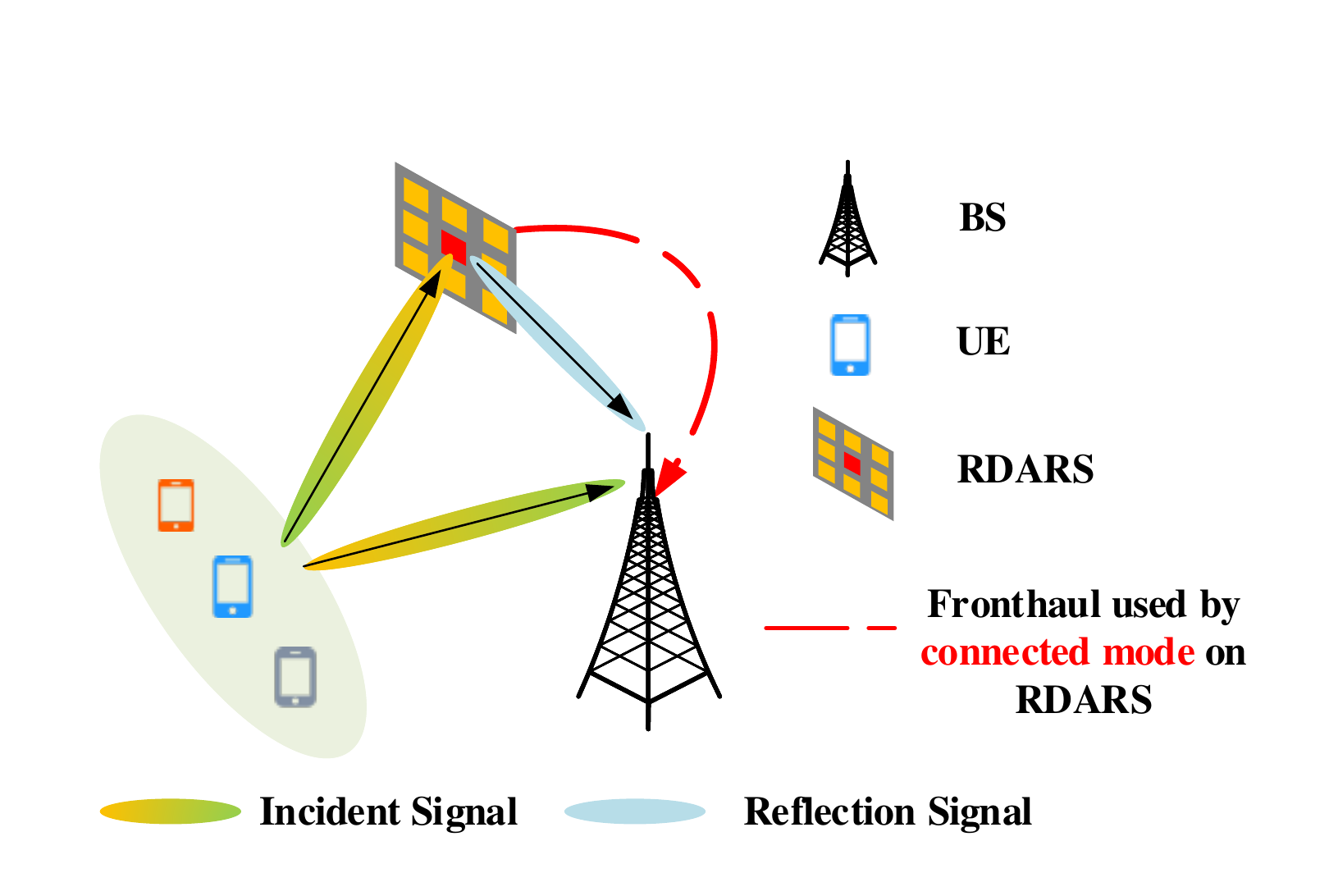} 
        } 
        \subfigure[downlink transmission] {
        \includegraphics[width=0.85\columnwidth]{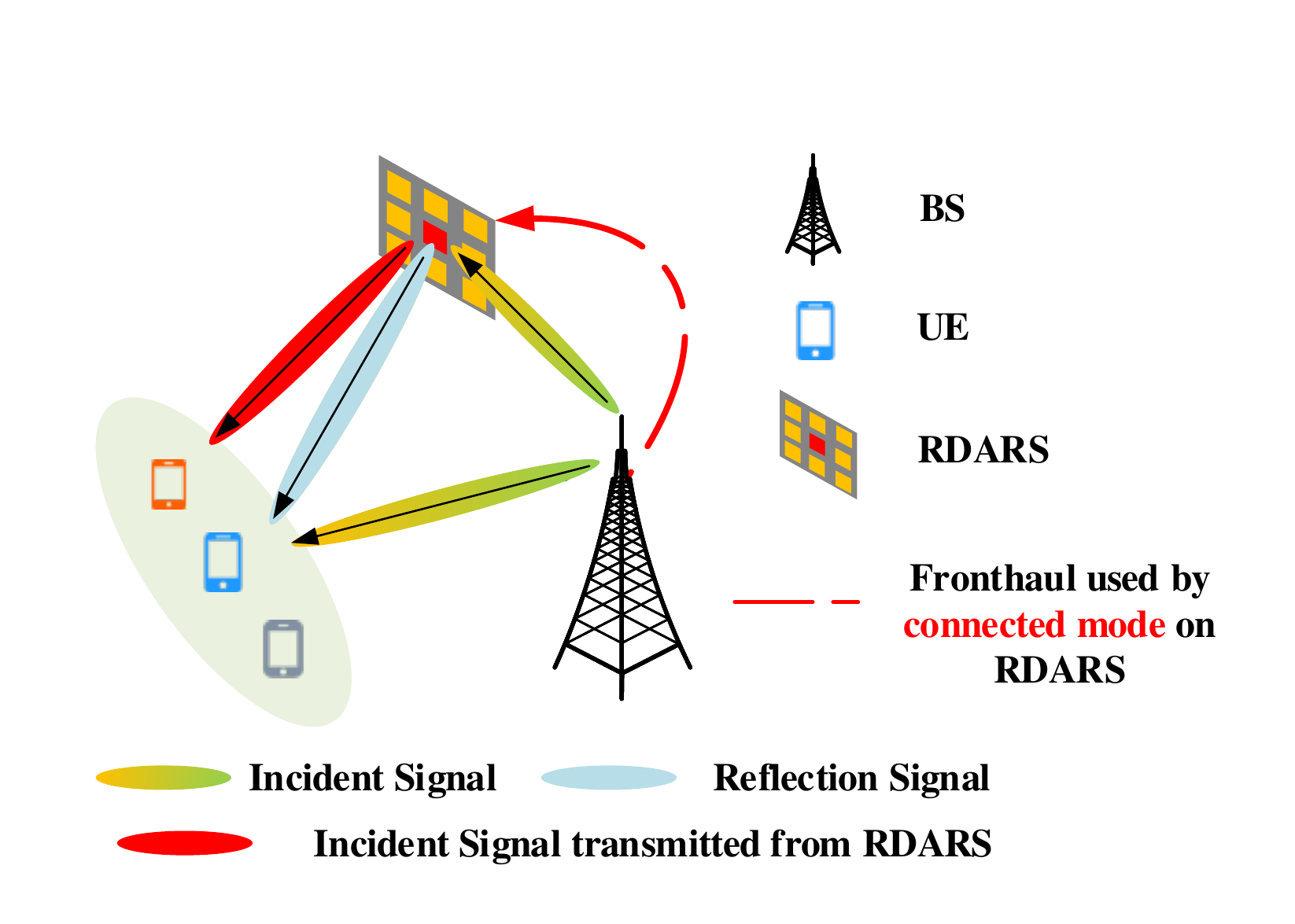} 
        }
        % \subfigure[An illustration of RDARS] {
        % \includegraphics[width=0.93\columnwidth]{./FigV2/RDARS_Architecture.pdf} 
        % }
        \caption{ 
        Illustration of RDARS-aided system. (a) The element acting \textit{connected mode} on RDARS is used for receiving signals. (b) The element acting \textit{connected mode} on RDARS is used for transmitting signals.  }    
    \end{figure*}

    % \subfigure[RDARS-aided system] {
    %     \includegraphics[width=0.5\columnwidth]{./FigV2/RDARS_Architecture.pdf} 
    %     }

    % \begin{figure} [htb] 
    % 	\centering
    %         % \vspace{-0.8cm}
    %         % \setlength{\belowcaptionskip}{-0.5cm}
    %         \includegraphics[width=0.5\textwidth]{./FigV2/RDARS_Architecture.pdf}
    % 	\caption{ A potential realization of RDARS: The \textit{connected mode} is built upon existing RRU (consists of the electrical-optical interface (E-O Conversion), the bandpass filter (BPF), and power amplifier (PA))  using  
    %     ROF technologies. \cite{Lisu_Yu}.  }
    % \end{figure}

\section{System Model}

    % \footnote{Here we only consider single RDARS for a starter since the system model of multiple RDARSs may involve multi-bounce signal reflected by the passive reflecting elements, which is still an open problem even for multi-RIS-aided systems\cite{Pingping_Zhang}.}
    In this paper, we first consider RDARS-aided uplink communications where an RDARS is deployed to assist the communications from a single-antenna UE to a single-antenna BS to shed light on the potential of such novel architecture as shown in Fig. 2(a). The multi-antenna BS scenario is then studied in Section V. The channels between RDARS-BS link, UE-RDARS, and UE-BS link are denoted as $\mathbf{h}_{RB} \in \mathbb{C}^{N \times 1}$, $\mathbf{h}_{UR} \in \mathbb{C}^{N \times 1}$, and $h_{UB} \in \mathbb{C}$, respectively. Armed with the notation of indicating matrix $\mathbf{A}$, the received signal $y \in \mathbb{C}$ from the antenna of BS and the received signal $\mathbf{u} \in \mathbb{C}^{N \times 1}$ from the $a$ elements acting with \textit{connected mode} at the RDARS can be expressed as\footnote{ We assume perfect fronthaul that offers the error-free and infinite capacity to characterize the performance upper bound for tractable insights similar to the analysis in DAS \cite{Lin_Dai} and cell-free mMIMO \cite{Ngo}. Future work is needed to quantify the impact of practical constraints.}

    % \begin{small}
    \begin{align} \label{system_model_SISO}
                \underbrace{\begin{bmatrix}
                    y^{T} \
                    \mathbf{u}^{T}
                \end{bmatrix}^{T}}_{\in \mathbb{C}^{(1+a) \times 1}} 
                &=
                \underbrace{\begin{bmatrix}
                    \overbrace{
                    h_{UB} + 
                    \mathbf{h}_{RB}^{H} \mathbf{B} \mathbf{h}_{UR}} ^ {\in \mathbb{C}^{1 \times 1}}  \\
                    \underbrace{ \mathbf{A} \mathbf{h}_{UR}}_{\in \mathbb{C}^{a \times 1}} 
                \end{bmatrix}}_{\in \mathbb{C}^{(1+a) \times 1}} \underbrace{x}_{\in \mathbb{C}} 
                +
                \underbrace{\begin{bmatrix}
                    n_{B}\\
                    \mathbf{n}_{R}
                \end{bmatrix}}_{\in \mathbb{C}^{(1+a) \times 1}}, 
                % = 
                % \widetilde{\mathbf{h}} x + \widetilde{\mathbf{n}},
    \end{align}
    % \end{small}
    
    \noindent where $y \in \mathbb{C}$ and $\mathbf{u} \in \mathbb{C}^{a \times 1}$ are the received signal at BS and at the elements acting \textit{connected mode} on RDARS, respectively, $x \in \mathbb{C}$ denotes the transmit symbol with total transmit power $\mathbb{E}[|x|^2] \leq P$. $n_{B} \sim \mathcal{CN}(0, \sigma_B^2)$ and $\mathbf{n}_R \sim \mathcal{CN}(\mathbf{0},\sigma_R^2\mathbf{I})$ are the additive white Gaussian noise (AWGN) received at BS and the element performing \textit{connected mode} on RDARS, respectively.
    \begin{remark}
     \begin{itemize}
         \hspace*{\fill}
         \item [(1)] For $a = 0$, the RDARS-aided system degrades into an RIS-aided system \cite{Qingqing_Wu}.
         \item [(2)] For $a = N$, the RDARS-aided system can be regarded as a DAS that has two distributed antenna sets, and one set has $N$ antennas located at a remote place while the other set of antennas is deployed at the BS \cite{Lin_Dai}. 
     \end{itemize}
    \end{remark}

    It can be seen from the above \textbf{Remark 1} that RDARS serves as a flexible and reconfigurable combination of DAS and RIS via different system parameters. To further demonstrate the potential of RDARS, the performance gain introduced by the use of RDARS will be characterized in the following sections.

    % \begin{figure} [htb]
    % 	\centering
    %         % \vspace{-0.8cm}
    %         % \setlength{\belowcaptionskip}{-0.5cm}
    %         \includegraphics[width=0.5\textwidth]{./FigV2/system_model.pdf}
    % 	\caption{ The RDARS-aided system.}
    % \end{figure} 

    \section{Performance Analysis}

    In this section, we reveal the potential of the proposed RDARS in terms of ergodic achievable rate. To this end, we first analyze the distribution of the received SNR and then derive the ergodic achievable rate. After that, we compare the performance of RDARS-aided systems with that of RIS-aided and DAS systems. To demonstrate the superiority of the proposed architecture and provide meaningful insights, we assume Rayleigh-fading channels for the above channel, i.e., $\mathbf{h}_{RB} \sim \mathcal{CN}(0, \beta^2\mathbf{I})$, $\mathbf{h}_{UR} \sim \mathcal{CN}(0, \alpha^2\mathbf{I})$, and $h_{UB} \sim \mathcal{CN}(0, \gamma^2)$ similar to \cite{Qingqing_Wu,Ngo}. 
    
    \subsection{Ergodic Achievable Rate}

    Recalling the system model in (\ref{system_model_SISO}), to effectively combine the received signal, a low-complexity combining scheme, i.e., the maximum ratio combining (MRC), is adopted at the BS.\footnote{The analysis can be easily extended to the RDARS-aided system with other combining schemes.} 
    We also assume $\sigma_{B} = \sigma_{R} = \sigma$ similar to the setting in DAS \cite{Lin_Dai} in this section.\footnote{ Such noise variance assumption is relaxed in Section V.} By applying the MRC and denoting the transmit SNR as $\overline{\gamma}_s \triangleq \frac{P}{\sigma^2}$, the received SNR $\gamma_s$ is given by 
    % \begin{small}
    \begin{align}  \label{gamma_s}
        \gamma_s 
        &=
        \overline{\gamma}_s ( | h_{UB} + \mathbf{h}_{RB}^{H} \mathbf{B} \mathbf{h}_{UR} |^2 + \mathbf{h}_{UR}^{H} \mathbf{A}^{H} \mathbf{A} \mathbf{h}_{UR} ).
    \end{align}
    % \end{small}
    
    \noindent It is clear that the indicator matrix $\mathbf{A}$ and the reflection-coefficient matrix $\mathbf{\theta}$ in $\mathbf{B}$ affect the value of $\gamma_s$. To unleash the potential of the RDARS system, the optimal $\boldsymbol{\theta}$ for RDARS can be derived as
    
    % Some insightful observations can be found by considering different system parameter configurations as follows:
    % \begin{itemize}
    %     \item  When $a = 0$ (or $\mathbf{A} = \mathbf{0}_N$), 
    % $\gamma_s$ becomes the received SNR of a fully-passive RIS-aided SISO system with a direct link \cite{Charishma}.
    %     \item When $a = N$ (or $\mathbf{A} = \mathbf{I}_N$),
    % $\gamma_s$ is the received SNR of a DAS with two distributed antenna sets where one set is the home-based station with one antenna and the other deployed at the remote RDARS consists of $N$ co-located antennas \cite{Junyuan_Wang}.
    % \end{itemize}

    \begin{theorem}
       The optimal phase shifts for elements acting \textit{reflection mode} on RDARS with perfect instantaneous CSI in the above system setup can be obtained in closed-form as: 
       % ${\rm{arg}}(\boldsymbol{\theta}(i)) =  {\rm{arg}} (h_{UB}) 
       %      + {\rm{arg}}( h_{RB,i} )
       %      - {\rm{arg}}( h_{UR,i} ),
       %      \forall \mathbf{A}^{H}\mathbf{A}(i,i) = 0$
        \begin{align}
            &{\rm{arg}}(\boldsymbol{\theta}(i)) = \notag 
            \\ 
            &{\rm{arg}} (h_{UB}) 
            + {\rm{arg}}( h_{RB,i} )
            - {\rm{arg}}( h_{UR,i} ),
            \forall \mathbf{A}^{H}\mathbf{A}(i,i) = 0,
        \end{align}
    where ${\rm{arg}}(\cdot) $ retrieves the phase of the input complex number. Also, $h_{RB,i}$ and $h_{UR,i}$ are the $i^{th}$ element of the channel vectors $\mathbf{h}_{RB}$ and $\mathbf{h}_{UR}$, respectively.
    \end{theorem}
    \begin{proof}
        The proof can be obtained by leveraging the triangle inequality similar to \cite[(24),(28)]{Qingqing_Wu} thus omitted here for brevity.
    \end{proof}
    \textbf{Theorem 1} indicates that the optimal phase shift design for elements performing \textit{reflection mode} is the same as the traditional RIS system\cite{Qingqing_Wu} and therefore does not require additional complexity in phase shift design in RDARS-aided systems when compared to RIS-aided systems.

    With the optimal phase shifts for RDARS\footnote{
    Here, $\mathbf{A}$ is assumed to be fixed to ease our analysis and give meaningful insights with fair comparison with RIS and DAS. However, the design of $\mathbf{A}$ provides an extra degree of freedom in optimization for RDARS-aided systems and will be investigated in the future.}, $\gamma_s$ in \eqref{gamma_s} can be further formulated as
    
    % \begin{small}
    \begin{align}
        &\gamma_s 
        =
        \overline{\gamma}_s 
        \times \notag \\
        &\bigg [ \underbrace{(|h_{UB}| +  \sum_{i=1}^{N}(1-a_i)|h_{RB,i}||h_{UR,i}| )^2}_{\textit{reflection gain}} 
        +
        \underbrace{\mathbf{h}_{UR}^{H} \mathbf{A}^{H} \mathbf{A} \mathbf{h}_{UR} )}_{\textit{distribution gain}}  \bigg ], \label{SNR_SISO}
    \end{align}
    % \end{small}
    
    \noindent where $a_i = \mathbf{A}^{H}\mathbf{A}(i,i) \in \{0, 1\}$. Here we define $(|h_{UB}| + |\mathbf{h}_{RB}^{H} \mathbf{B} \mathbf{h}_{UR}| )^2$ as the $\textit{reflection gain}$ from the elements in RDARS performing \textit{reflection mode}, while $\mathbf{h}_{UR}^{H} \mathbf{A}^{H} \mathbf{A} \mathbf{h}_{UR}$ as the $\textit{distribution gain}$ provided by the elements in RDARS working under \textit{connected mode}. 
    
    It can be seen from \eqref{SNR_SISO} that $\gamma_s$ involves the sum of random variables (RVs). Specifically, the $\textit{reflection gain}$ in $\gamma_s$ provided by the passive reflecting elements and the direct link contains the square of the sum of RVs, i.e., $(|h_{UB}| +  \sum_{i=1}^{N}(1-a_i)|h_{RB,i}||h_{UR,i}| )^2$. Therefore, obtaining the exact distribution of $\gamma_s$ involves multiple integrals, convolution, and transformation of RVs and thus makes  $\gamma_s$ intractable. Besides, the analysis is different from the power scaling law derived for RIS-aided and active RIS-aided systems without the consideration of direct link in \cite{Qingqing_Wu, Zijian_Zhang}, where the Law of Large Number (LLN) is used for analysis. Existing analytical results can not be directly applied due to the involvement of direct link $h_{UB}$ and the introduced distribution gain from the connected-mode elements in \eqref{SNR_SISO}, thus making the distribution derivation of $\gamma_s$ significantly challenging. To address this issue, we first propose the following proposition based on Gamma moments matching \cite{Van_Chien, Charishma} to characterize the distribution of $\gamma_s$, and then the ergodic achievable rate is derived.

    \begin{proposition}
        The distribution of $\gamma_s$ can be approximated as
            % \begin{small}
            \begin{align}
                \gamma_s \ \overset{approx}{\sim} \Gamma(k,
                p
                ), \label{gamma approx}
            \end{align} 
            % \end{small}
            
        \noindent with the shape and scale parameters as
        
        % \begin{small}
        \begin{align}
            k = \frac{\mathbb{E}[\gamma_s]^2}{ \mathbb{E}[\gamma_s^2] - \mathbb{E}[\gamma_s]^2}, \
            p = \frac{\mathbb{E}[\gamma_s^2] - \mathbb{E}[\gamma_s]^2}{\mathbb{E}[\gamma_s]},
        \end{align}
        % \end{small}
        
        \noindent where the first and second moments of $\gamma_s$, i.e., $\mathbb{E}[\gamma_s]$ and $\mathbb{E}[\gamma_s^2]$, can be evaluated 
        from \eqref{EX} and \eqref{EX2} at the top of the next page, respectively.
    \end{proposition}
    \begin{proof}
        See Appendix A.
    \end{proof}

    %   E^{signal}_appendix
    \newcounter{Expectation}
        \begin{figure*}[!htb]
        \small
        % \hrulefill
        % Ensure that we have normal size text
        % \normalsize
        % Store the current equation number.
        \setcounter{Expectation}{\value{equation}}
        \setcounter{equation}{6}
        \begin{align} 
         \mathbb{E}[\gamma_s] \label{EX}  
            &=
            \overline{\gamma_s} \bigg[
            N^2 
                    \big(
                    \frac{\pi^2}{16}\alpha^2\beta^2
                    \big)
                    +
                    \alpha^2\beta^2 
                    \big( 
                    N(1 - \frac{\pi^2}{8}a - \frac{\pi^2}{16}) + \frac{\pi^2}{16}a(a+1)-a
                    \big)
                    +
                    \alpha \beta \gamma
                    \big(
                    \frac{\pi}{4}\sqrt{\pi}(N-a)
                    \big)
                    +
                    a\alpha^2 + \gamma^2 \bigg], 
            \\
            %%%
        \mathbb{E}[\gamma_s^2] \label{EX2}
            &=
            \overline{\gamma}_s^2 [
            2\gamma^{4} 
            +
            6 (N-a)(1+ \frac{\pi^2}{16}(N-a-1)) \gamma^2 \alpha^{2} \beta^{2}
            +
            \frac{3\pi}{4} \sqrt{\pi}  (N-a) \gamma^{3} \alpha \beta 
            + 
            C_4 \alpha^{4} \beta^{4} 
            +
            2\sqrt{\pi}  C_3 \gamma^3 \beta^3 \gamma 
            + 
            2 a \gamma^2 \alpha^2 \notag \\ 
            &+ 
            2 (N-a)(1+ \frac{\pi^2}{16}(N-a-1))a \alpha^{4} \beta^{2} 
            +
            \frac{\pi}{2}\sqrt{\pi} (N-a) a \alpha^3 \gamma \beta
            +
            a(a+1) \alpha^4
            ]. 
        \end{align}
        % Restore the current equation number.
        \setcounter{equation}{\value{Expectation}}
        % The IEEE uses it as a separator
        \hrulefill
        \end{figure*}
        \addtocounter{equation}{2}

    \subsubsection{Ergodic Achievable Rate}
    Since $\gamma_s$ is approximated as a Gamma RV, a tractable ergodic achievable rate can be obtained as
    
        % \begin{small}
        \begin{align}
            R
            &= \mathbb{E}[(\log_2(1+\gamma_s))],
            \\
            &=
            \frac{1}{\Gamma(k){\rm{ln}}(2)} 
            H_{3, 2}^{1, 3}
            \bigg[
            p \bigg |  
                    \begin{matrix}
                        (1,1) ,& (1,1) ,& (-k+1,1) \\
                        (1,1) ,& (0,1)
                    \end{matrix}
            \bigg] , \label{EAC Gamma}
        \end{align}
        % \end{small}
        
        \noindent where 

        % \begin{small}
        \begin{align}
            H(z) = H_{p, q}^{m, n}
            \bigg[
                p \bigg |  
                    \begin{matrix}
                        (a_i,A_i)_{1,p} \\
                        (b_i,B_i)_{1,q}
                    \end{matrix}
            \bigg],
        \end{align}
        % \end{small}
        
        \noindent is the Fox's H-function (FHF) defined via the Mellin-Barnes type integral \cite[(2)]{Fox1}.
        \begin{proof}
        Defining $\frac{1}{p}\gamma_s = y$, we have
        $R 
            =
            \int_{0}^{\infty} \log_2(1+\gamma_s) f_{\gamma_s}(\gamma_s) d\gamma_s
            =
            \frac{1}{\Gamma(k)} \int_{0}^{\infty} y^{k-2} \log_2(1+py) \ y  \ {\rm{exp}}(-y) dy 
            \overset{a}{=}
            \frac{1}{\Gamma(k) {\rm{ln}}(2)}
            % \times 
            \int_{0}^{\infty} y^{k-2} 
            H_{2, 2}^{1, 2}
            \bigg[
                py \bigg |  
                    \begin{matrix}
                        (1,1) & (1,1) \\
                        (1,1) & (0,1)
                    \end{matrix}
            \bigg]
            \
            H_{0, 1}^{1, 0}
            \bigg[
                y \bigg |  
                    \begin{matrix}
                        -- \\
                        (1,1)
                    \end{matrix}
            \bigg]
            dy$, where $\overset{a}{=}$ is derived by expressing ${\rm{log}_2}(1+py)$ and $y {\rm{exp}}(-y)$ as FHF. After applying the integral involving the product of two FHFs from \cite[(2.8.4)]{Fox2}, we obtain \eqref{EAC Gamma}.
        \end{proof}
        
        % \begin{align}
        %     R 
        %     &=
        %     \int_{0}^{\infty} \log_2(1+\gamma_s) f_{\gamma_s}(\gamma_s) d\gamma_s
        %     \notag
        %     \\
        %     &=
        %     \frac{1}{\Gamma(k)} \int_{0}^{\infty} y^{k-2} \log_2(1+py) \ y  \ {\rm{exp}}(-y) dy 
        %     \notag
        %     \\
        %     &\overset{a}{=}
        %     \frac{1}{\Gamma(k) {\rm{ln}}(2)}
        %     % \times 
        %     \int_{0}^{\infty} y^{k-2} 
        %     H_{2, 2}^{1, 2}
        %     \bigg[
        %         py \bigg |  
        %             \begin{matrix}
        %                 (1,1) & (1,1) \\
        %                 (1,1) & (0,1)
        %             \end{matrix}
        %     \bigg]
        %     \
        %     H_{0, 1}^{1, 0}
        %     \bigg[
        %         y \bigg |  
        %             \begin{matrix}
        %                 -- \\
        %                 (1,1)
        %             \end{matrix}
        %     \bigg]
        %     dy. 
        % \end{align}

    % The Gamma moments matching is shown to be accurate by well capturing the probability mass of $\gamma_s$, and further leads to accurate approximation in (\ref{EAC Gamma}) which is numerically validated in Section VI. 
    
    \subsubsection{Upper Bound}
    
    By leveraging Jensen's inequality and $E[\gamma_s]$ in \eqref{EX}, an upper bound of the ergodic achievable rate can be derived as expressed in \eqref{EAC upper bound} at the top of the next page.
            % \begin{align}
            %     R^{U} =  \notag \label{EAC upper bound}
            %     \log_2 \bigg[ & 1
            %     +
            %     \overline{\gamma}_s
            %     \bigg(
            %     N^2 
            %     \big(
            %     \frac{\pi^2}{16}\alpha^2\beta^2
            %     \big)
            %     +
            %     \alpha^2\beta^2 
            %     \big( 
            %     N(1 - \frac{\pi^2}{8}a^2 - \frac{\pi^2}{16}) + \frac{\pi^2}{16}(a^2 + a - 1)
            %     \big) \\
            %     &+
            %     \gamma \alpha \beta
            %     \big(
            %     \frac{\pi}{4}\sqrt{\pi}(N-a)
            %     \big)
            %     +
            %     a\alpha^2 + \gamma^2
            %     \bigg)
            %      \bigg].
            % \end{align}
    Such upper bound is also shown to be accurate with large $N$ which will be validated numerically in Section VI.

        %   E^{signal}_appendix
    \newcounter{EAC_upperbound}
        \begin{figure*}[!htb]
        \small
        % Ensure that we have normal size text
        % \normalsize
        % Store the current equation number.
        \setcounter{EAC_upperbound}{\value{equation}}
        \setcounter{equation}{11}
        \begin{align} 
        R^{U} &=  \label{EAC upper bound}
                \log_2 \bigg[  1
                +
                \overline{\gamma}_s
                \bigg(
                N^2 
                \big(
                \frac{\pi^2}{16}\alpha^2\beta^2
                \big)
                +
                \alpha^2\beta^2 
                \big( 
                N(1 - \frac{\pi^2}{8}a^2 - \frac{\pi^2}{16}) + \frac{\pi^2}{16}(a^2 + a - 1)
                \big)
                +
                \gamma \alpha \beta
                \big(
                \frac{\pi}{4}\sqrt{\pi}(N-a)
                \big)
                +
                a\alpha^2 + \gamma^2
                \bigg)
                 \bigg].
        \end{align}
        % Restore the current equation number.
        \setcounter{equation}{\value{EAC_upperbound}}
        % The IEEE uses it as a separator
        \hrulefill
        \end{figure*}
        \addtocounter{equation}{1}

    \begin{remark} \label{Remark 1}
        It can be seen that, when N is large, the upper bound $R^{U}$ in \eqref{EAC upper bound} is scaling up with $\mathcal{O}(\log_2(N^2))$ which coincides with the power scaling law for conventional RIS \cite{Qingqing_Wu}. 
        This indicates that the RDARS-aided system (using MRC) can effectively acquire the benefits from the passive elements as the corresponding RIS-aided system in the asymptotic regime (with optimal phase shifts), i.e., $N \rightarrow \infty$. Note that the SNR of the recently proposed active RIS-aided system in \cite{Zijian_Zhang} only scales linearly with $\mathcal{O}(\log_2(N))$ in the asymptotic regime due to the additional noises introduced by the use of active components.
    \end{remark}

    % Moreover, as shown in the Fig. 2 and the simulation results provided in Section \uppercase\expandafter\romannumeral5\uppercase\expandafter{\romannumeral5}, the RDARS-aided system will perform better than the RIS-aided system for the practical number of NN by compensating the ``multiplicative fading'' effect with distribution gain\textit{distribution gain} provided by the elements performing \textit{connected mode}.  These emphasize the practical application value of the proposed RDARS architecture.

    \subsection{Comparison between RDARS, RIS, and DAS}

    As indicated from {\textbf{Remark\,\ref{Remark 1}}}, the RDARS inherits the advantages of conventional RIS in terms of \textit{reflection gain} when $N$ is large. Apart from the \textit{reflection gain}, the \textit{distribution gain} will dominate in the RDARS for a limited $N$. This can be utilized to efficiently alleviate the ``multiplicative fading'' \cite{Zijian_Zhang} encountered by the conventional RIS and enhance the performance of conventional DAS via low-cost components simultaneously. The following corollary specifies the condition that the RDARS-aided system outperforms its DAS and RIS-aided system counterparts. 

    \begin{corollary}
        \begin{itemize}
            \hspace*{\fill}
            \item [(1)] For any given $a$ ($0<a <N$), the RDARS-aided system always outperforms its DAS counterpart.
            \item [(2)] For any given $a$ ($0<a<N$), the RDARS-aided system will outperform its RIS-aided system counterpart if the total number of elements satisfies
                % \begin{small}
                \begin{align} \label{Lemma1_2}
                    0< N < \frac{8}{\pi^2}(\frac{1}{\beta^2} -1) + \frac{a+1}{2} - \frac{2\sqrt{\pi}}{\pi} \frac{\gamma}{\alpha \beta}. 
                \end{align}
                % \end{small}
        \end{itemize}        
    \end{corollary}
    \begin{proof}
        See Appendix B.
    \end{proof}

    % From the above \textbf{Corollary 3} (1), it can be observed that the condition in (?????????????????????\ref{Lemma1_1}) can be easily met as long as αβ≪γ\alpha\beta \ll \gamma, i.e., the product of UE-RDARS and RDARS-BS large-scale fading coefficients is much smaller than the UE-BS link path-loss coefficient (which is usually satisfied in the practical scenario). This indicates that the RDARS will outperform DAS in most scenarios.

    It is clear that \textbf{Corollary 3}. (1) coincides with physical interpretation by direct inspection of \eqref{SNR_SISO}, i.e., $N-a$ elements acting \textit{reflection mode} on RDARS coherently align the phases of reflection link and direct UE-BS link, and thus provide higher \textit{reflection gain} with term $\sum_{i=1}^{N}(1-a_i)|h_{RB,i}||h_{UR,i}|$ in \eqref{SNR_SISO}. For (\ref{Lemma1_2}), considering the scenario without UE-BS link, we will have $0< N < \frac{8}{\pi^2}(\frac{1}{\beta^2} -1) + \frac{a+1}{2}$. which means the RDARS will outperform the RIS counterpart once $N$ is less than the upper bound, i.e., $\frac{8}{\pi^2}(\frac{1}{\beta^2} -1) + \frac{a+1}{2}$. Note that this upper bound is generally a tremendous value since $\beta$ is relatively small.
    
    To give a clear illustration, we provide a detailed example in the following. Let the transmit SNR be $\overline{\gamma}_s = 90$ dB, $\gamma^2 = \beta^2 = \alpha^2 = -70$ dB\cite{Zijian_Zhang}, and $a = 1$. It can be shown that the upper bound for $N$ in (\ref{Lemma1_2}) is about $8.1 \times 10^6$. The average received SNRs $\mathbb{E}[\gamma_s]$ of the RDARS, the RIS-aided system, and the DAS are presented in Fig. 3 where the results without RDARS are also provided as ``W.O. RDARS".
    
    \begin{figure} [htb] \label{Fig2}
    	\centering
            % \vspace{-0.8cm}
            % \setlength{\belowcaptionskip}{-0.5cm}
            \includegraphics[width=0.5\textwidth]{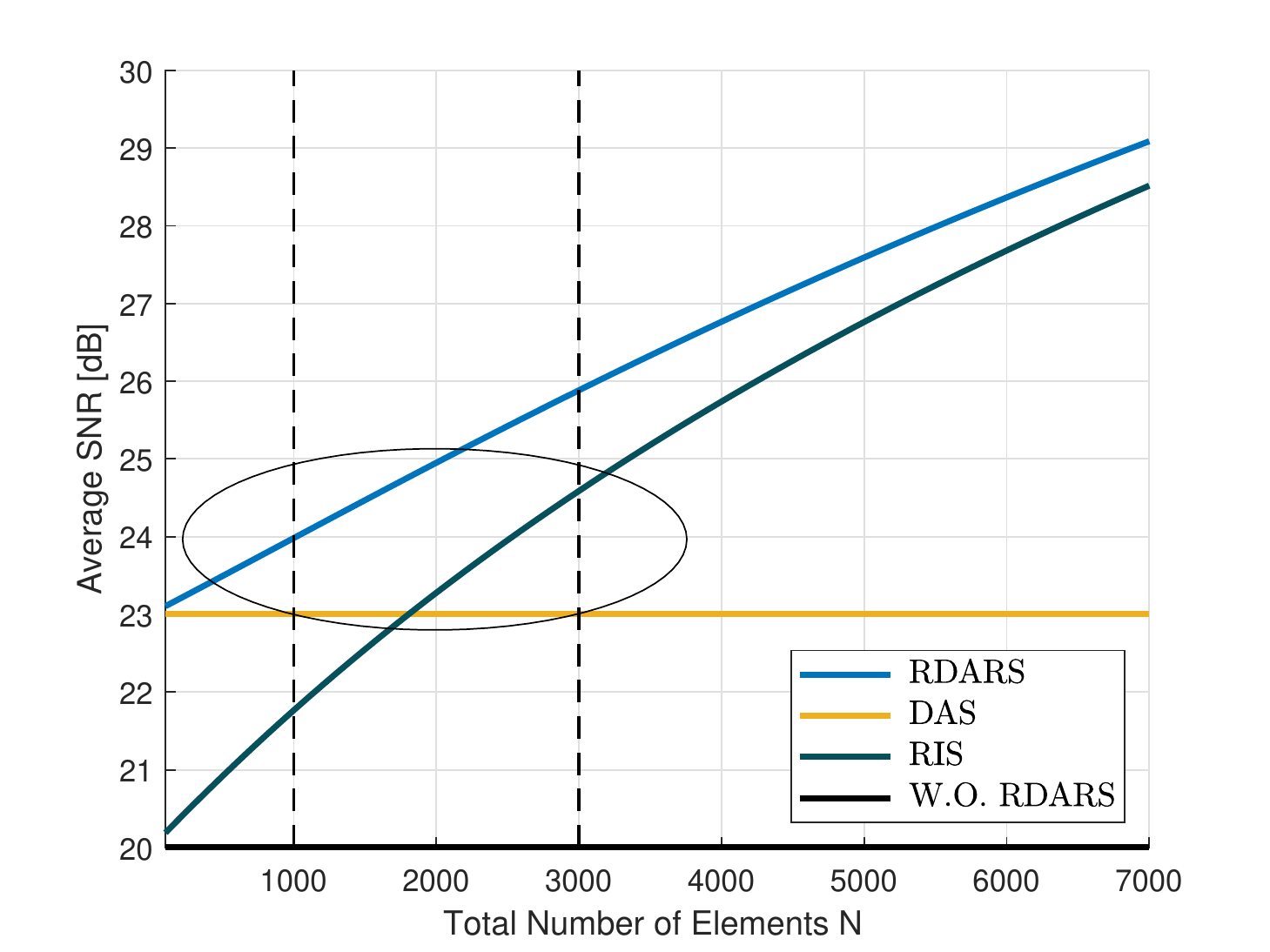}
    	\caption{ An example for comparison between different systems with only one element acting \textit{connected mode} on RDARS, i.e., $a = 1$.}
    \end{figure}

    \begin{remark} \label{Remark 2}
        The merits of RDARS have been clearly demonstrated via the comparison in Fig. 3. For a small $N$, RDARS mainly leverages the \textit{distribution gain} and behaves as DAS and thus combats the ``multiplicative fading'' effect of conventional RIS. For large $N$, RDARS greatly outperforms DAS by exploiting the \textit{reflection gain} inherited from RIS. Most importantly, for practical $N$, e.g., $N \in [1000, 2000]$ in this system setup, RDARS can strike a balance between \textit{reflection gain} and \textit{distribution gain} and, therefore, achieves the highest average received SNR. More simulation results will be presented in Section VI to support this insight. 
    \end{remark}

    \section{Extension to Multi-antenna BS scenario}

    In this section, we extend the performance analysis of the proposed RDARS-aided system into multi-antenna BS scenarios.

    % To account for more practical settings, different from the previous section, both the multi-antenna BS scenario and the LoS channel are taken into account in this section.
    % Although insightful observation can be found when assuming instantaneous CSI available and optimal RDARS configuration in (???????????????????????????????????????????????????????????????????????????????????????????????????????????????????????????????????????\ref{SNR_SISO}) for analysis in the previous section when considering single-antenna BS scenario similar to \cite{Qingqing_Wu, Zijian_Zhang}, such assumption may not be practical in deployment due to large overhead introduced by frequent acquisition of CSI as well as controlling of surface, especially with large NN. As such, we focus on unveiling the potential of RDARS without stringent assumptions of phase shift design for RDARS configuration in this section.
    
    % \footnote{The statistical CSI refers to path-loss parameter, Rician factor as well as the AoAs and AoDs in this paper, which can be easily obtained. Such statistical characteristics could keep invariant within the considered time period and thus can be utilized for system design to avoid high overhead from frequent configuration. To keep derivation tractable, perfect statistical CSI is assumed similar to \cite{Kangda_Zhi2}, while an extension incorporating imperfect CSI will be our future work.} 

    \subsection{System Model}
    Consider an uplink RDARS-aided system where a single-antenna UE communicates with a $L$-antenna BS.
    Let $x \in \mathbb{C}$ be the signal transmitted from the UE to BS with total power $\mathbb{E}[|x|^2] \leq P$, the received signal $\mathbf{y} \in \mathbb{C}^{L \times 1}$ from BS and the received signal $\mathbf{u} \in \mathbb{C}^{a \times 1}$ from the $a$ ($a \leq N$) elements acting with \textit{connected mode} at the RDARS can be expressed as
    
    % \begin{small}
    \begin{align} \label{system_model_SIMO}
                \underbrace{\begin{bmatrix}
                    \mathbf{y}^{T} \
                    \mathbf{u}^{T}
                \end{bmatrix}^{T}}_{\in \mathbb{C}^{(L+a) \times 1}} 
                &=
                    \underbrace{\begin{bmatrix}
                        \overbrace{ 
                        \mathbf{H} \mathbf{B} \mathbf{h} + \mathbf{d} }^ {\in \mathbb{C}^{L \times 1}}  \\
                        \underbrace{ \mathbf{A} \mathbf{h} }_{\in \mathbb{C}^{a \times 1}} 
                    \end{bmatrix}}_{\triangleq \widetilde{\mathbf{h}} \in \mathbb{C}^{(L+a) \times 1}} \underbrace{x}_{\in \mathbb{C}} 
                    +
                    \underbrace{\begin{bmatrix}
                        \mathbf{n}_{B}\\
                        \mathbf{n}_{R}
                    \end{bmatrix}}_{\triangleq \widetilde{\mathbf{n}} \in \mathbb{C}^{(L+a) \times 1}} 
                    = 
                    \widetilde{\mathbf{h}} x + \widetilde{\mathbf{n}}, 
    \end{align}
    % \end{small}
    
    \noindent where $\mathbf{H} \in \mathbb{C}^{L\times N}$, $\mathbf{h} \in \mathbb{C}^{N \times 1}$, and $\mathbf{d} \in \mathbb{C}^{N \times 1}$ denote the channel matrix (or vector) for RDARS-BS link, UE-RDARS link, and UE-BS link, respectively. $\mathbf{n}_{B} \sim \mathcal{CN}(0, \sigma_{B}^2\mathbf{I})$ and $\mathbf{n}_{R} \sim \mathcal{CN}(0, \sigma_{R}^2\mathbf{I})$ are AWGN at BS and the elements performing \textit{connected mode} on RDARS, respectively.

    \subsection{Channel Model}
    To incorporate the LoS propagation between UE-RDARS and RDARS-BS, we adopt the Rician fading channel model for UE-RDARS and RDARS-BS links \cite{Kangda_Zhi2}. As such, the RDARS-BS and UE-RDARS channel models are expressed as 
    % $\mathbf{H} 
    % =
    %     \sqrt{\frac{\beta}{\delta+1}} (\sqrt{\delta} \overline{\mathbf{H}} + \widetilde{\mathbf{H}})$,
    % $\mathbf{h} 
    % = 
    %     \sqrt{\frac{\alpha}{\epsilon+1}} (\sqrt{\epsilon} \overline{\mathbf{h}} + \widetilde{\mathbf{h}})$,
    % $\mathbf{d} 
    % =
    %    \sqrt{\gamma} \widetilde{\mathbf{d}}$,    
    \begin{align}
    \mathbf{H} 
    &=
        \sqrt{\frac{\beta}{\delta+1}} (\sqrt{\delta} \overline{\mathbf{H}} + \widetilde{\mathbf{H}}),
    \\
    \mathbf{h} 
    &= 
        \sqrt{\frac{\alpha}{\epsilon+1}} (\sqrt{\epsilon} \overline{\mathbf{h}} + \widetilde{\mathbf{h}}),
    \end{align}
    where $\delta$ and $\epsilon$ are the Rician factors of the RDARS-BS and UE-RDARS link, respectively. $\overline{\mathbf{H}} \in \mathbb{C}^{L \times N}$ and $\overline{\mathbf{h}} \in \mathbb{C}^{N \times 1}$ denote the LoS components of the RDARS-BS and UE-RDARS links, respectively. In this paper, we adopt uniform planar array (UPA) for both BS and RDARS with the size as $L = L_x \times L_y$ and $N = N_x \times N_y$. Accordingly, we have
    $\overline{\mathbf{H}} = \mathbf{a}_L(\phi_{RB}^{a}, \phi_{RB}^{e}) 
            \mathbf{a}_N^{H}(\varphi_{RB}^{a}, \varphi_{RB}^{e})$,
    $\overline{\mathbf{h}} = \mathbf{a}_N(\varphi_{UR}^{a}, \varphi_{UR}^{e})$,
    % \begin{align}
    %     \overline{\mathbf{H}} &= \mathbf{a}_L(\phi_{RB}^{a}, \phi_{RB}^{e}) 
    %         \mathbf{a}_N^{H}(\varphi_{RB}^{a}, \varphi_{RB}^{e}),
    %     \\
    %     \overline{\mathbf{h}} &= \mathbf{a}_N(\varphi_{UR}^{a}, \varphi_{UR}^{e}), 
    % \end{align}
    where $\varphi_{UR}^{a}$ ($\varphi_{UR}^{e}$) is the azimuth (elevation) angle of the arrival (AoA) of the incident signal at the RDARS from UE, $\varphi_{RB}^{a}$ ($\varphi_{RB}^{e}$) is the azimuth (elevation) angle of the departure (AoD) reflected by the RDARS towards BS, and $\phi_{RB}^{a}$ ($\phi_{RB}^{e}$) is the azimuth (elevation) AoA of the signal received at BS from the RDARS, respectively. Besides, $\mathbf{a}_{X}(\psi^a, \psi^e) \in \mathbb{C}^{X \times 1}$ with $ X \in \{ L, N \}$ and $ \psi \in \{ \phi_{RB}, \varphi_{RB}, \varphi_{UR}  \}$ denote the array response vector, whose x-entry is

    \begin{small}
    \begin{align}
        [\mathbf{a}_{X}(\psi^a, \psi^e)](x) 
        &= {\rm{exp}} \bigg\{ 
            j2\pi\frac{d_s}{\lambda}
        \bigg(
                \lfloor (x-1)/X_y \rfloor \sin \psi^a \sin \psi^e
        \notag
        \\
        &+
                ((x-1){\rm{mod}}X_y) \cos \psi^e
            \bigg)
        \bigg\},
    \end{align}
    \end{small}
    
    \noindent where $d_s$ and $\lambda$ denote the element spacing of the antenna array, and the wavelength, respectively.

    \subsection{Ergodic Achievable Rate}

    Recalling from the system model in (\ref{system_model_SIMO}), the signal after applying MRC can be obtained as  
   \begin{align}
            r &=
                \widetilde{\mathbf{h}}^{H} \widetilde{\mathbf{h}} x + \widetilde{\mathbf{h}}^{H} \widetilde{\mathbf{n}} .
    \end{align}
    Thus, the ergodic achievable rate can be given by
    \begin{align} \label{EAC_SIMO}
        R = \mathbb{E} \bigg\{ \log_2 \bigg( 1 + \frac{P || \widetilde{\mathbf{h}}||^4}{ \widetilde{\mathbf{h}}^{H}\mathbf{R}\widetilde{\mathbf{h}} } \bigg)  \bigg\},
    \end{align}
    where $\mathbf{R} = {\rm{blk}}(\sigma_B^2 \mathbf{I}_L, \sigma_R^2 \mathbf{I}_a) \in \mathbb{C}^{(L+a) \times (L+a)}$ is the noise covariance matrix. Note that it is challenging to derive an exact closed-form expression for (\ref{EAC_SIMO}) due to the expectation outside the logarithmic operation. As such, we resort to \cite[{\rm{Lemma 1}}]{Qi_Zhang} to obtain an closed-form approximation of the ergodic achievable rate as 
    \begin{align} \label{EAC_SIMO_appro}
        R \approx \bigg\{ \log_2 \bigg( 1 + \frac{P \mathbb{E}[|| \widetilde{\mathbf{h}}||^4]}{ \mathbb{E}[\widetilde{\mathbf{h}}^{H}\mathbf{R}\widetilde{\mathbf{h}} ] } \bigg)  \bigg\}.
    \end{align}
    It is clear that (\ref{EAC_SIMO_appro}) only requires the calculation of the moments of $\widetilde{\mathbf{h}}$. However, the derivations of these moments are challenging due to the involvement of the Rician fading model for both UE-RDARS and RDARS-BS links. 
    % As a result, even for conventional RIS-aided system, the derivatives are challenging \cite{Kangda_Zhi2}, let alone considering a more general RDARS-aided system.\footnote{The system model of RDARS can be reduced to RIS-aided system following a similar approach as mentioned in the previous section III.} 
    Nevertheless, by tackling the above challenges, we provide the closed-form approximate expression of the ergodic achievable as follows.

    \begin{theorem}
        The ergodic achievable rate of the RDARS-aided system can be approximated as
        % \begin{small}
        \begin{align} \label{EAC_SIMO_approx}
            R \approx \log_2 \bigg( 
                1   +   \frac{PE^{signal}(\mathbf{A},\mathbf{\Theta})}{E^{noise}(\mathbf{A},\mathbf{\Theta})}
            \bigg),
        \end{align}
        % \end{small}
        
        \noindent where the definitions of $E_{k}^{signal}$ and $E_{k}^{noise}$ are provided in (\ref{E^{signal}}) and (\ref{E^{noise}}), at the top of the next page, respectively. Also, $c \triangleq \frac{\beta\alpha}{(\delta+1)(\epsilon+1)}$, $d \triangleq \frac{\alpha}{\epsilon + 1}$, and finally, $f(\mathbf{A},\mathbf{\Theta}) \in \mathbb{C}$ with 
        \begin{align}
            f(\mathbf{A},\mathbf{\Theta}) \triangleq \sum_{n=1}^{N} (1-a_n) e^{j(\zeta_n + {\rm{arg}}(\boldsymbol{\theta}_n))},
        \end{align}
        where $\zeta_n = 2\pi\frac{d}{\lambda}\big( \lfloor (n-1)/\sqrt{N} \rfloor \big( {\rm{sin}}(\varphi_{UR,k}^e){\rm{sin}}(\varphi_{UR,k}^a) - {\rm{sin}}(\varphi_{RB}^e){\rm{sin}}(\varphi_{RB}^a) \big) + \big( (n-1){\rm{mod}}\sqrt{N} \big) \big({\rm{cos}}(\varphi_{UR,k}^e) - {\rm{cos}}(\varphi_{RB}^e) \big) 
        \big)$.
    \end{theorem}
    \begin{proof}
        See Appendix C.
    \end{proof}

    %   E^{signal}
    \newcounter{E^{signal}}
        \begin{figure*}[!htb]
        \small
        % ensure that we have normalsize text
        % \normalsize
        % Store the current equation number.
        \setcounter{E^{signal}}{\value{equation}}
        \setcounter{equation}{22}
        \begin{align} \label{E^{signal}}
            &E^{signal} =
                L c^2 \bigg\{ 
                    L(\delta\epsilon)^2 |f(\mathbf{A},\mathbf{\Theta})|^4
                    +
                    2 \delta \epsilon |f(\mathbf{A},\mathbf{\Theta})|^2
                \times
                    (
                        2L(N-a)\delta + L(N-a)\epsilon + L(N-a) + 2L + (N-a)\epsilon + N-a + 2
                    )
                \notag
                \\
                &+
                    L(N-a)^2 (
                            2\delta^2 + \epsilon^2 + 2\delta\epsilon + 2\delta + 2\epsilon + 1
                        )
                +
                    (N-a)^2 (
                        \epsilon^2 + 2 \delta\epsilon + 2\delta + 2\epsilon + 1
                    )
                +
                    L(N-a)(2\delta + 2\epsilon + 1)
                +
                    (N-a)(2\delta+2\epsilon+1)
                \bigg\}
                \notag
                \\
                &+
                    Lc
                    \bigg\{ 
                        2 \delta\epsilon|f(\mathbf{A},\mathbf{\Theta})|^2  [(L+1)\gamma + a \alpha ]
                    +
                        2 (N-a)(L+1)(\epsilon+\delta+1) \gamma
                    \bigg\}
                +
                    2 L a \bigg\{ 
                       (N-a) c \alpha (\epsilon + \delta + 1)
                       +
                       \alpha \gamma
                    \bigg\}
                \notag
                \\
                &
                +
                    L(L+1)\gamma^2
                +
                    a d^2 \bigg\{
                        (\epsilon+1)^2 a + 2\epsilon + 1
                    \bigg\}.
        \end{align}
        % Restore the current equation number.
        \setcounter{equation}{\value{E^{signal}}}
        % The IEEE uses as a separator
        % \hrulefill
        \end{figure*}
        \addtocounter{equation}{1}

    %   E^{noise}
    \newcounter{E^{noise}}
        \begin{figure*}[!htb]
        \small
        % ensure that we have normalsize text
        % \normalsize
        % Store the current equation number.
        \setcounter{E^{noise}}{\value{equation}}
        \setcounter{equation}{23}
        \begin{align} \label{E^{noise}}
            E^{noise}(\mathbf{A},\mathbf{\Theta})
        &=    
               \sigma_B^2 L \big[
                    |f(\mathbf{A},\mathbf{\Theta}) |^2 c \delta \epsilon
                    +
                    (N-a) c \delta 
                    +
                    \big( 
                        (N-a)c(\epsilon + 1) + \gamma
                    \big) 
                \big]
            +
               \sigma_R^2 a 
               d (\epsilon + 1).
        \end{align}
        % Restore the current equation number.
        \setcounter{equation}{\value{E^{noise}}}
        % The IEEE uses as a separator
        \hrulefill
        \end{figure*}
        \addtocounter{equation}{1}

    Notice that the evaluation of the ergodic achievable rate in (\ref{EAC_SIMO_approx}) has significantly lower computational complexity as compared to time-consuming Monte Carlo simulations. Moreover, such closed-form expression characterizes the impacts of $L$, $N$, AoA and AoD, path-loss parameters, Rician factors as well as RDARS configuration, i.e., $\mathbf{A}$ and $\mathbf{\Theta}$, on the system performance. This has been exemplified in the following remark.

    \begin{remark}
        \begin{itemize}
            \item [(1)] It can be easily proved that for given $a$ ($a \leq N$), any $\mathbf{A}$ satisfies ${\rm{Tr}}(\mathbf{A}^{H}\mathbf{A}) = a$ and $\mathbf{A}^{H}\mathbf{A}(i,i) \in \{0, 1 \}$ would obtain the same ergodic achievable rate in (\ref{EAC_SIMO_approx}). As a result, arbitrary selection of the elements as \textit{connected mode} is sufficient for the single UE uplink case considered in this paper. 
            \item [(2)] It is observed that $|f(\mathbf{A},\mathbf{\Theta})| \leq N-a$ where equality holds when ${\rm{arg}}(\boldsymbol{\theta}(n)) = \zeta_n$ by leveraging the triangle inequality similar to \cite[(24),(28)]{Qingqing_Wu}. We refer to the phase shift design of RDARS configuration that satisfies $|f(\mathbf{A},\mathbf{\Theta})| = N-a$ when $a < N$ as ``\textit{phase shifts aligned to user}" similar to \cite{Kangda_Zhi2}. When such RDARS configuration is applied, $E^{signal}(\mathbf{A},\mathbf{\Theta})$ and $E^{noise}(\mathbf{A},\mathbf{\Theta})$ scale with $\mathcal{O}(L^2N^4)$ and $\mathcal{O}(LN^2)$, respectively. As a result, the ergodic achievable rate of RDARS-aided system in (\ref{EAC_SIMO_approx}) achieves a performance gain of $\mathcal{O}(\log_2(LN^2))$. This indicates that the RDARS exploits the reflection gain inherited from RIS for large $N$ as reported in \cite{Kangda_Zhi2}.     
        \end{itemize}
    \end{remark}

    Notice that the obtained expression (\ref{EAC_SIMO_approx}) for the RDARS-aided system is also applicable to that of the corresponding RIS-aided system and DAS counterparts. Therefore, we utilize it in the following section to compare the performance of these systems.

    \section{Simulation Results}
    In this section, numerical results are presented to validate the effectiveness of the proposed architecture and the correctness of the performance analysis. 

    \subsection{Simulation Setup}

    In the simulation, the log-distance path loss models from the 3GPP standards \cite[Table B.1.2.1-1]{3GPP} are adopted as
    \begin{align}
        PL_{i} &= C_0 + 10\alpha_{i}\log_{10}d + z_i, i \in \{ UB, UR, RB \},
    \end{align}
    where $C_0$ is the path loss at reference distance, $\alpha_i$ is the path loss exponent, and $z_i$ is shadow fading following Gaussian distribution with standard deviation $\sigma_{sh}$. The 3D-cartesian coordinate is adopted and the BS, RDARS, and UE are located as (0m, 0m, 10m), (20m, 20m, 10m), and (200m, 0m, 1.5m). Unless specified otherwise, the corresponding parameters are set as $C_0 = 30$, $\{ \alpha_{UR}, \alpha_{RB} , \alpha_{UB} \} = \{2.5, 2.0, 3.1 \}$, and $\sigma_{sh} = 3$ dB as \cite{Ngo}. Furthermore, the transmit power and noise power are set as $P = 10$ dBm and $\sigma^2 = \sigma_B^2 = \sigma_R^2 =  -80$ dBm. For the scenario considering the LoS component, the Rician factors are set as $\delta = 10$ and $\epsilon = 10$, unless specified otherwise. The AoA and AoD are generated randomly from [0,2$\pi$] \cite{Kangda_Zhi2} and are fixed once generated.
    % \footnote{In the simulation, the angles are set as ($\phi_{RB}^{a}, \varphi_{RB}^{a}, \varphi_{UR}^a) = (6.28, 4.17, 5.20)$, ($\phi_{RB}^{e}, \varphi_{RB}^{e}, \varphi_{UR}^{e}) = (4.21, 0.09, 4.32)$. } 
    All the theoretical results are verified via Monte-Carlo simulations, each running $10000$ channel realizations. The ``Sim" and ``Alt" in the following figures correspond to the results obtained from the Monte-Carlo simulation and the closed-form expressions derived in the previous sections, respectively. 

    % \footnote{Here the RDARS is deployed near the BS to lower the cost for controlling as well as element operating at \textit{connected mode}.}

    \subsection{Single-Antenna BS Scenario}

    In this subsection, we verify the correctness of the analysis derived in Section IV and study the performance of the single-antenna BS scenarios. To demonstrate the tightness of the proposed approximation based on Gamma moment matching in Section IV, the lower bound based on the ergodic achievable rate derived for the RDARS-aided system via the method from \cite[(54)]{Bound} is also provided as shown in Fig. 4(a) for comparison.

    \begin{figure} [htb]  
        \centering
        \setlength{\belowcaptionskip}{-3mm}
        \subfigure[Rate versus total number of $N$] {
        \includegraphics[width=0.9\columnwidth]{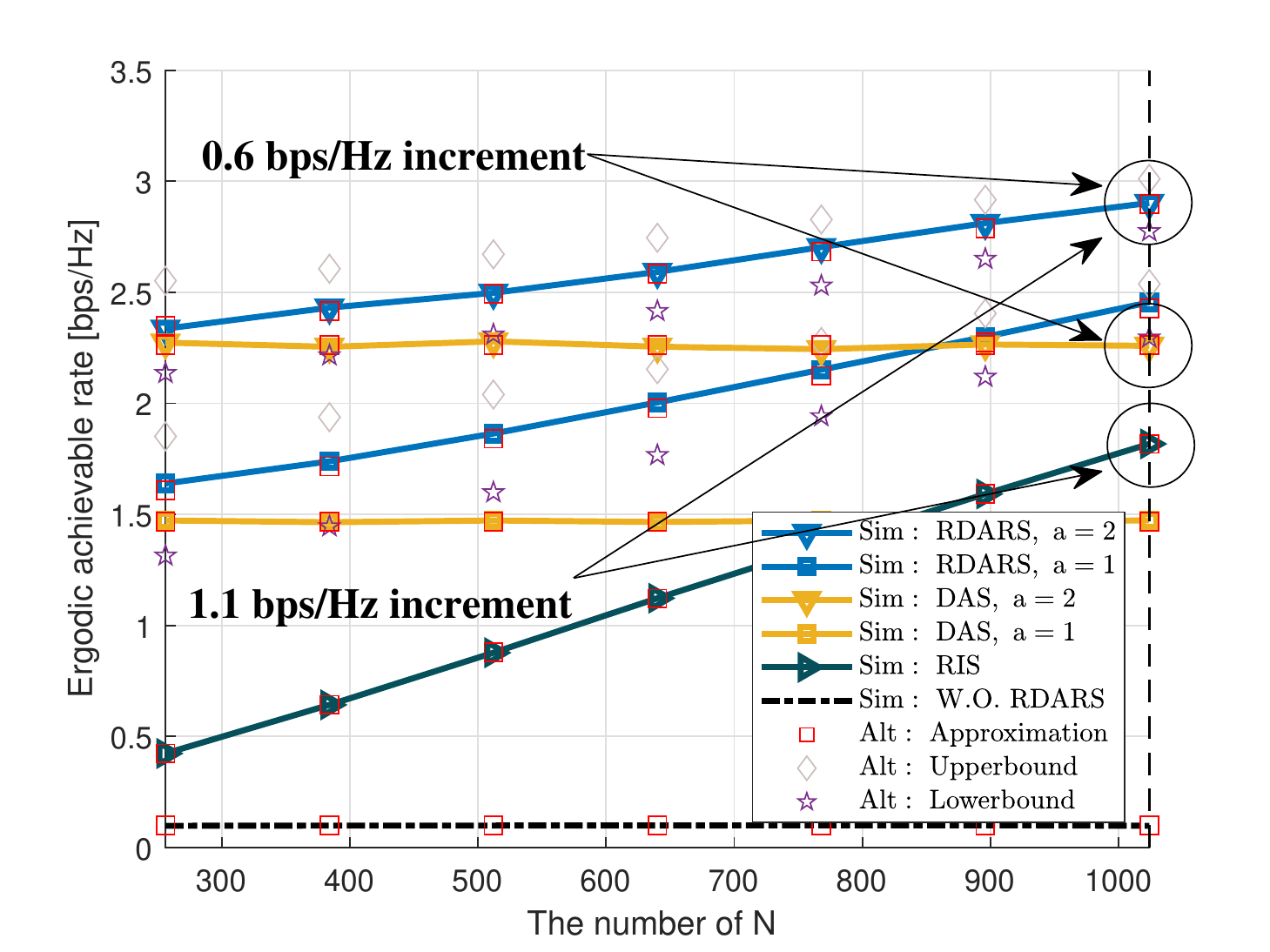}
        }  
        \\
        \subfigure[Rate versus transmit power $P$] {
        \includegraphics[width=0.9\columnwidth]{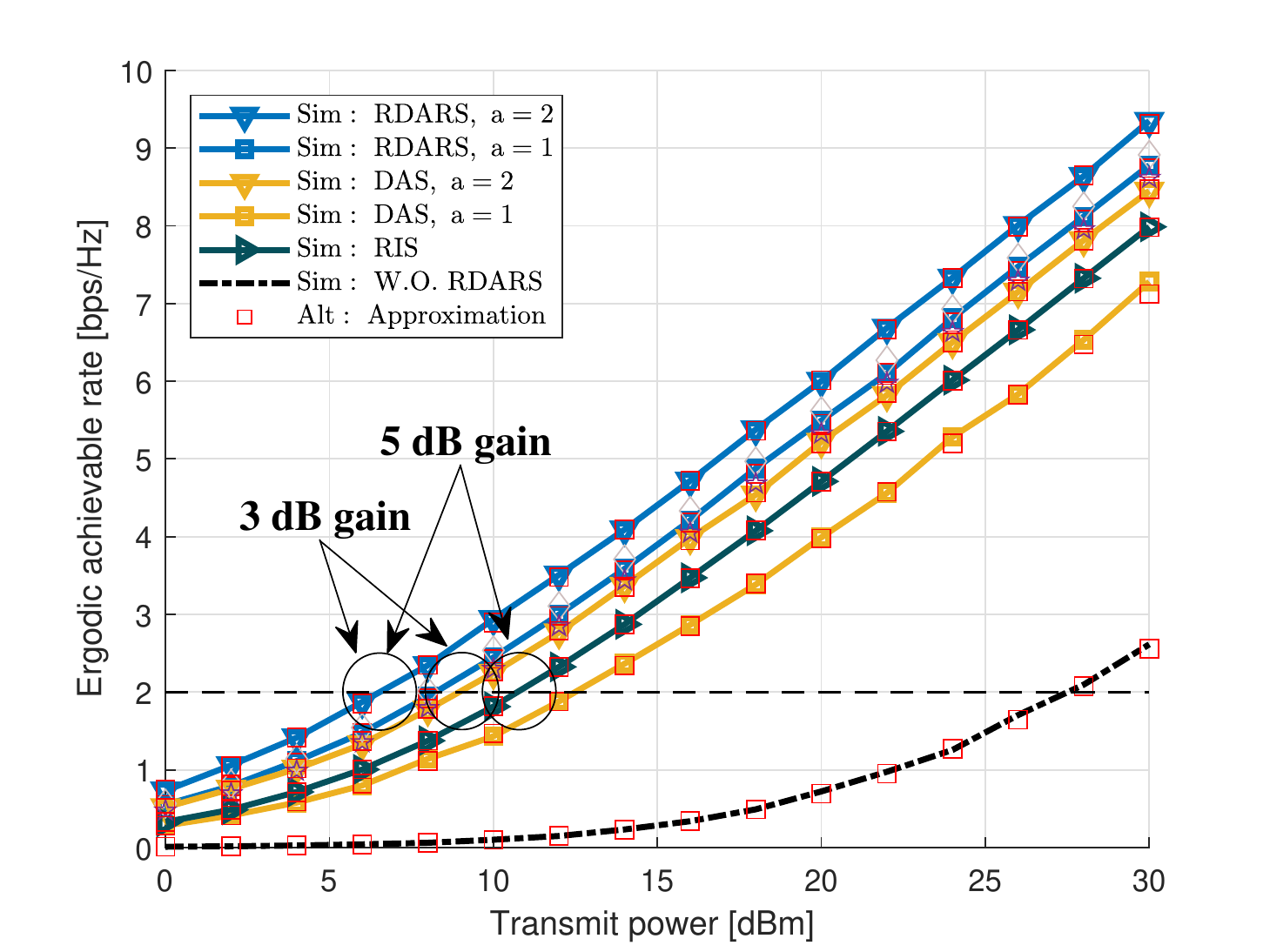} 
        }  
        \caption{ 
        The ergodic achievable rate under single-antenna BS scenarios. $L=1$.} 
    \end{figure}

    % \begin{figure} [htb]
    % 	\centering
    %         % \vspace{-0.8cm}
    %         % \setlength{\belowcaptionskip}{-0.5cm}
    %         \includegraphics[width=0.5\textwidth]{./FigV2/N_SISO}
    % 	\caption{ Rate versus total number of $N$ under single-antenna BS scenarios. $L=1$.}
    % \end{figure}

    % \begin{figure} [htb]
    % 	\centering
    %         % \vspace{-0.8cm}
    %         % \setlength{\belowcaptionskip}{-0.5cm}
    %         \includegraphics[width=0.5\textwidth]{./FigV2/P_SISO}
    % 	\caption{ Rate versus transmit power of $P$ under single-antenna BS scenarios. $L=1$.}
    % \end{figure}
    
    In Fig. 4, we compare the performance of the RDARS-aided, RIS-aided system, and the DAS systems with different numbers of $N$, $a$, and the transmit power $P$. 
    Results show that the derived ergodic achievable rate approximation \eqref{EAC Gamma} matches the simulation results under various $N$. Moreover, from Fig. 4(a), it can be observed that the lower \cite[(54)]{Bound} and the upper bounds \eqref{EAC upper bound} become tighter as both $N$ and $a$ increase. These, thus, verify the correctness of the derived ergodic achievable rate based on the approximation in \eqref{EAC Gamma}, and the effectiveness of the upper bound in \eqref{EAC upper bound} for large $N$.

    As for the system performance comparison, it can be seen that as the number of $N$ increases, the RDARS-aided system and RIS-aided system achieve higher ergodic achievable rates. Nevertheless, the RDARS-aided system significantly outperforms both DAS and RIS-aided systems under the practical setting of $N$, e.g., $N \leq 1024$\cite{Hanzo, Wankai_Tang, Xilong_Pei}. Specifically, for DAS ($a = 2$) and RIS-aided system ($N = 1024$), the rates are $2.3$ bps/Hz and $1.8$ bps/Hz, respectively, while RDARS-aided system with $N = 1024$ and $a = 2$ can achieve up to $2.9$ bps/Hz rate, thus acquiring additional $0.6$ bps/Hz and $1.1$ bps/Hz rate improvements. Moreover, as illustrated in \textbf{Remark 2} and \textbf{Remark 3}, RDARS outperforms its counterparts by simultaneously exploiting both \textit{distribution gain} and \textit{reflection gain} with a controllable trade-off for performance enhancement.
    % (a similar phenomenon is also observed in the experimental result presented in Section VII). 
    
    As can be seen from Fig. 4(b), the RDARS-aided system outperforms the DAS, the RIS-aided system, and the system without the assistance of RDARS with the increasing transmit power. This implies that the power consumption for the RDARS-aided system is much lower than that of the DAS and the RIS-aided system given the same ergodic achievable rate requirement. For example, to achieve a $2$ bps/Hz ergodic achievable rate, the required transmit powers for DAS and RIS-aided systems are about $9$ dBm and $11$ dBm, respectively, while the RDARS-aided system ($N = 1024$) with $a = 2$ only requires $6$ dBm. This demonstrates a $3$ dB and $5$ dB power-saving provided by RDARS-aided systems compared to DAS and RIS-aided systems, respectively.     

    % \subsubsection{Location of UE versus rate}

    % \begin{figure} [htb]
    % 	\centering
    %         % \vspace{-0.8cm}
    %         % \setlength{\belowcaptionskip}{-0.5cm}
    %         \includegraphics[width=0.5\textwidth]{./FigV2/Y_SISO}
    % 	\caption{ Rate versus UE's y-axis, $N = 1024$.}
    % \end{figure}

    % In Fig. X, we plot the ergodic achievable rate versus the different locations of UE, i.e., the UE location is (200m, $y$m, 1.5m) which varies with different $y$. First, it is observed that the derived results match well with the simulation results, these confirm the correctness of the derived results. Second, it is observed that the RDARS-aided system outperforms DAS and RIS-aided systems at a higher rate by leveraging both \textit{distribution gain} and \textit{reflection gain}.

    \subsection{Multi-Antenna BS Scenario}
    
    In this section, we verify the correctness of the analysis derived in Section V and study the performance of the RDARS-aided system under multi-antenna BS scenarios. Here we adopt ``phase shifts aligned to the user" for the phase shift design as mentioned in \textbf{Remark 4} \cite{Kangda_Zhi2} and set $\mathbf{\Theta} = \mathbf{\Theta}_a$. We also provide the result when $\mathbf{\Theta} = \mathbf{I}$ for comparison.\footnote{Though such design for $\mathbf{\Theta}$ is not global optimal, it still serves as an efficient solution as demonstrated in \cite{Kangda_Zhi2} and its numerical results are also given for comparison in this section.} The number of BS antennas is set as $L = 4$ and the total number of elements at RDARS is $N = 512$, unless otherwise specified.

    \subsubsection{Impacts of $N$, $a$, $L$ and $P$}
    
    \begin{figure} [htb]  
        \centering
        \setlength{\belowcaptionskip}{-3mm}
        \subfigure[Rate versus total number of $N$] {
        \includegraphics[width=0.9\columnwidth]{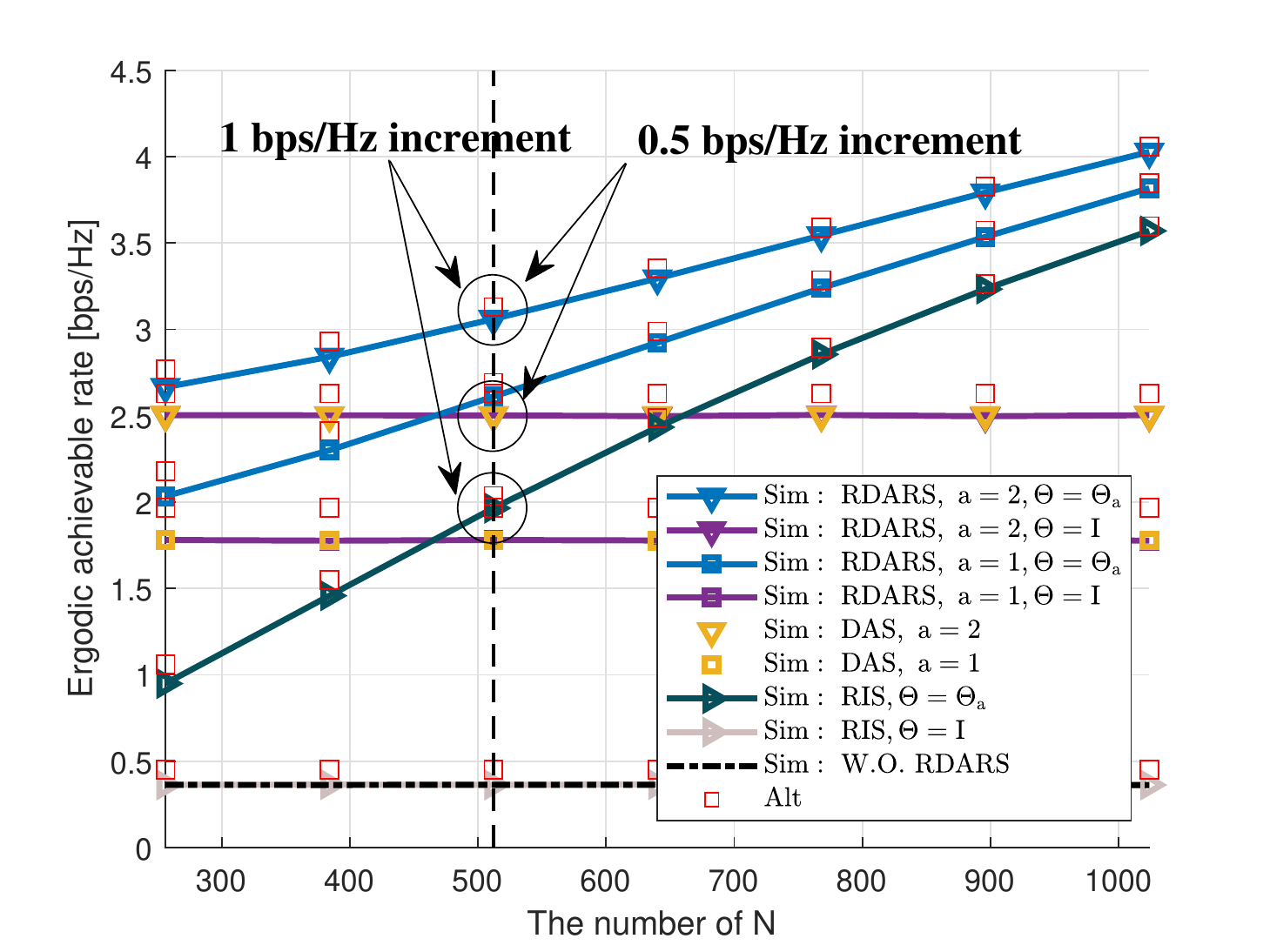}
        } 
        \\
        \subfigure[Rate versus transmit power $P$] {
        \includegraphics[width=0.9\columnwidth]{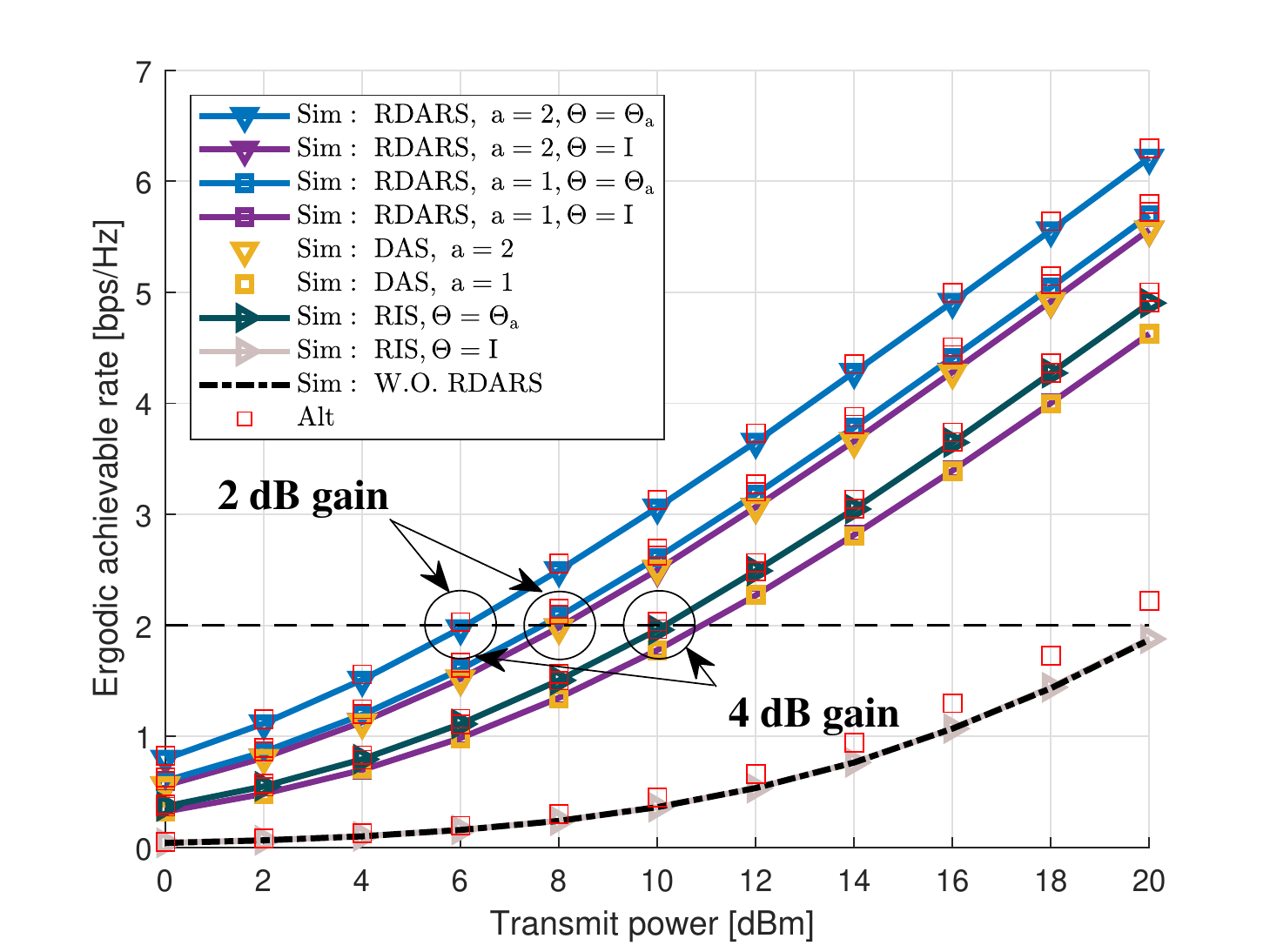} 
        }  
        \caption{ 
        The ergodic achievable rate under multi-antenna BS scenarios. $L=4$.} 
    \end{figure}

    % \begin{figure} [htb]
    % 	\centering
    %         % \vspace{-0.8cm}
    %         % \setlength{\belowcaptionskip}{-0.5cm}
    %         \includegraphics[width=0.5\textwidth]{./FigV2/N_SIMO}
    % 	\caption{ Rate versus total number of $N$ under single-antenna BS scenarios. $L=4$.}
    % \end{figure}

    % \begin{figure} [htb]
    % 	\centering
    %         % \vspace{-0.8cm}
    %         % \setlength{\belowcaptionskip}{-0.5cm}
    %         \includegraphics[width=0.5\textwidth]{./FigV2/P_SIMO}
    % 	\caption{ Rate versus transmit power $P$ under single-antenna BS scenarios. $L=4$.}
    % \end{figure}

    In Fig. 5(a), the performance of RDARS-aided, RIS-aided systems, and DAS are compared under different numbers of $N$ and $a$. First, as shown in Fig. 5(a), the closed-form ergodic achievable rate approximation derived in \eqref{EAC_SIMO_approx} matches well with the simulation results. Second, under the phase shift design of $\mathbf{\Theta}_a$, the RDARS-aided systems outperform the DAS and RIS-aided systems. A $0.5$ bps/Hz rate improvement has been attained when adopting $\mathbf{\Theta} = \mathbf{\Theta}_a$ for RDARS-aided system as compared to $\mathbf{\Theta} = \mathbf{I}$. Besides, an $1$ bps/Hz improvement is obtained by the RDARS as compared to the RIS counterpart, which thus demonstrates the merits of the proposed RDARS. In Fig. 5(b), the ergodic achievable rate is provided under different transmit power $P$. It is evident that by deploying RDARS, the transmit power can be reduced as compared to the counterparts. Therefore, this demonstrates the proposed RDARS's superiority by leveraging both \textit{distribution gain} and \textit{reflection gain}.  

    % Though there is a slight mismatch when evaluating DAS performance, \eqref{EAC_SIMO_approx} still serves as a good approximation for characterizing the system performance.

    \begin{figure} [!htb]  
        % \small
        % \setlength{\belowcaptionskip}{-3mm}  
        \centering
        \includegraphics[width=0.9\columnwidth]{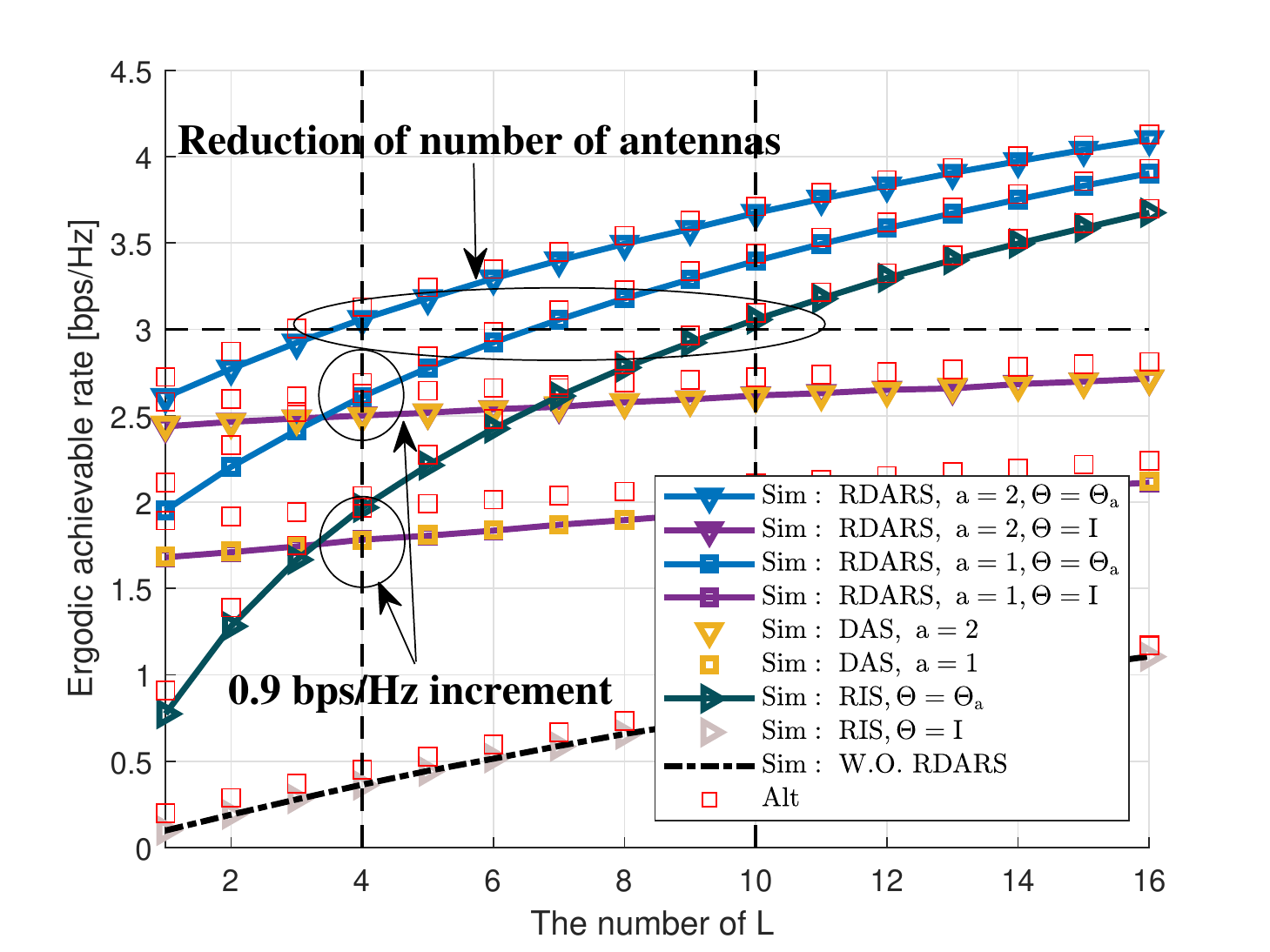} 
        \caption{ 
         The ergodic achievable rate versus the BS antennas $L$. $N = 512$.}     
    \end{figure}

    Fig. 6 investigates the impact of the number of BS antennas $L$ on the system ergodic achievable rate. 
    By leveraging elements performing reflection mode, the RDARS-aided system with $L = 4$ and $a = 1$ outperforms the DAS systems with $4$ antennas at the BS and single remote antenna by $0.9$ bps/Hz rate improvement. 
    Besides, to achieve the same performance as the RDARS-aided system with total $L+a = 6$ active antennas, e.g., $3$ bps/Hz, the RIS-aided system requires $10$ active antennas at the BS. Consequently, the RDARS-aided systems are promising for future DAS systems by reducing the required number of active antennas/RF chains, hardware cost, and energy consumption for DAS deployment while maintaining a satisfactory communication performance.

    \subsubsection{Impact of Rician Factor $\delta$}

    \begin{figure} [!htb]  
        % \small
        % \setlength{\belowcaptionskip}{-3mm}  
        \centering
        \includegraphics[width=0.9\columnwidth]{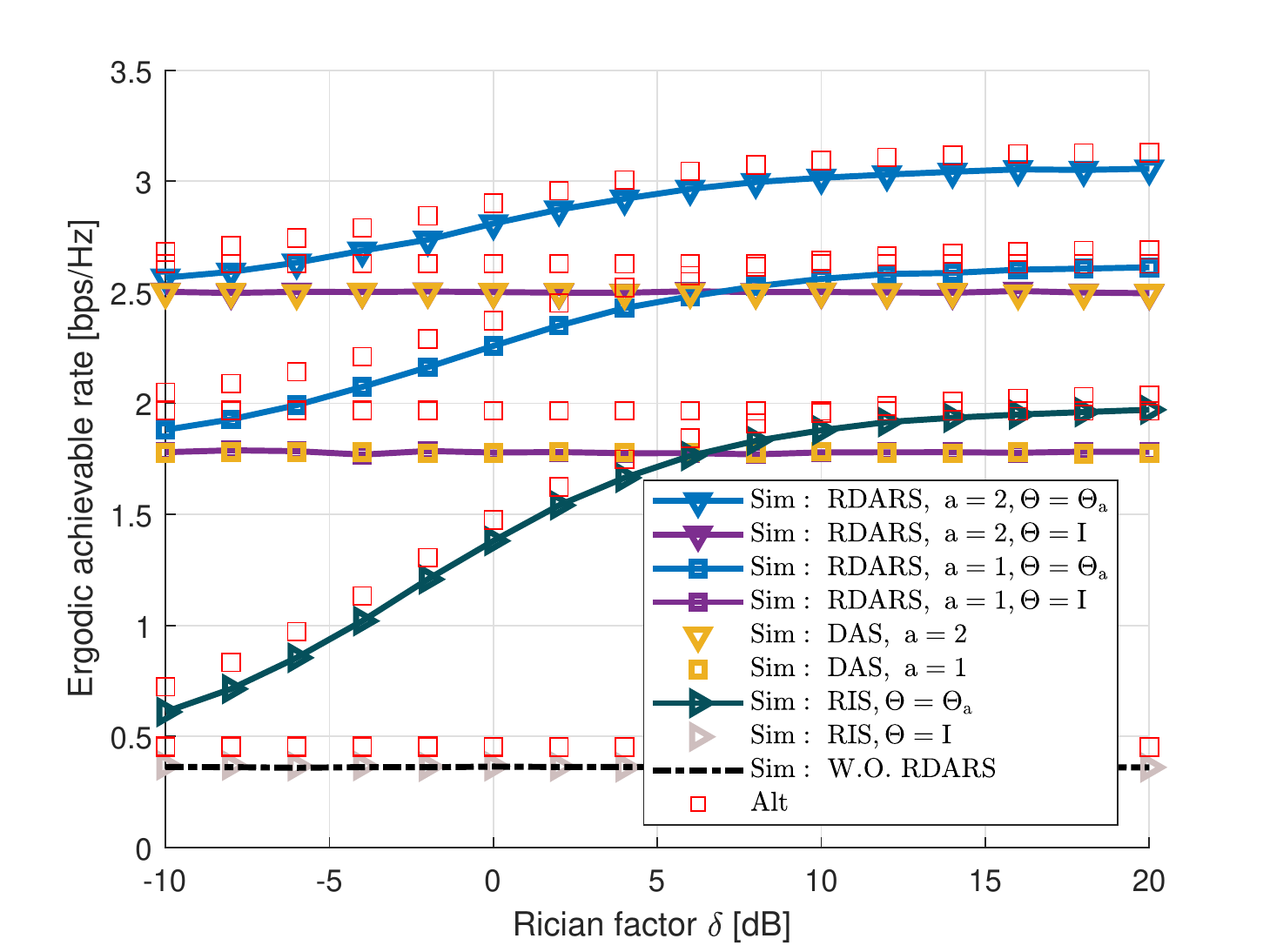} 
        \caption{ 
        The ergodic achievable rate versus Rician factor $\delta$. $N = 512$.}     
    \end{figure}

    In Fig. 7, the impact of the Rician factor $\delta$ on the system performance is investigated. When $\delta$ is small, the LoS path between the BS and RDARS is weak, and the \textit{reflection gain} is negligible for both RDARS and RIS-aided systems. Hence, the performance of the RDARS-aided system is almost identical to the DAS systems with a prominent \textit{distribution gain}. Notice that the performance of the RDARS-aided and RIS-aided systems increases with $\delta$. This is because the phase shift design of the elements acting as \textit{reflection mode} is aligned with the statistical CSI, i.e., $\mathbf{\Theta} = \mathbf{\Theta}_a$. As a result, with increasing $\delta$, the RDARS-BS link becomes strong and contributes positively, and a higher rate is acquired by leveraging the \textit{reflection gain} \cite{MingMin_Zhao}. This further verifies the superiority of the proposed RDARS in terms of the great flexibility in benefiting from the advantages of DAS and RIS-aided systems in different scenarios.    

    \section{Experimental result}
    
    In this section, experimental results are presented using a fabricated prototype of RDARS to verify the performance of this proof-of-concept.
    Fig. 8(a) presents the picture of the proposed RDARS-aided wireless communication system. The system comprises one transmitter (UE), one receiver (BS), and one RDARS. The transmitter and receiver consist of a computer, a Universal Software Radio Peripheral (USRP), and an antenna.
    % \footnote{Note that only one antenna is connected to the USRP at the receiver (BS) as shown in Fig. 10.(b).}
    The detailed system parameters are presented in TABLE I.

    %  \begin{figure} [htb]  
    %     \centering
    %     \setlength{\belowcaptionskip}{-3mm}
    %     % \subfigure[UE (Transmitter)] {
    %     % \includegraphics[width=0.47\columnwidth]{./FigV1/TX.pdf}
    %     % }  
    %     % \subfigure[BS (Receiver)] {
    %     % \includegraphics[width=0.44\columnwidth]{./FigV1/RX.pdf} 
    %     % }   
    %     % \\
    %     % \subfigure[System setup] {    
    %     % \includegraphics[width=1\columnwidth]{./FigV1/system_setup.pdf} 
    %     % }
    %     \centering
    %         % \vspace{-0.8cm}
    %         % \setlength{\belowcaptionskip}{-0.5cm}
    %         \includegraphics[width=1\columnwidth]{./FigV1/system_setup.pdf}
    %     \caption{ 
    %    Illustration of RDARS-aided wireless communication system.} 
    % \end{figure}

    \begin{figure} [htb]  
        \centering
        \setlength{\belowcaptionskip}{-3mm}
        \subfigure[RDARS-aided system] {    
        \includegraphics[width=1\columnwidth]{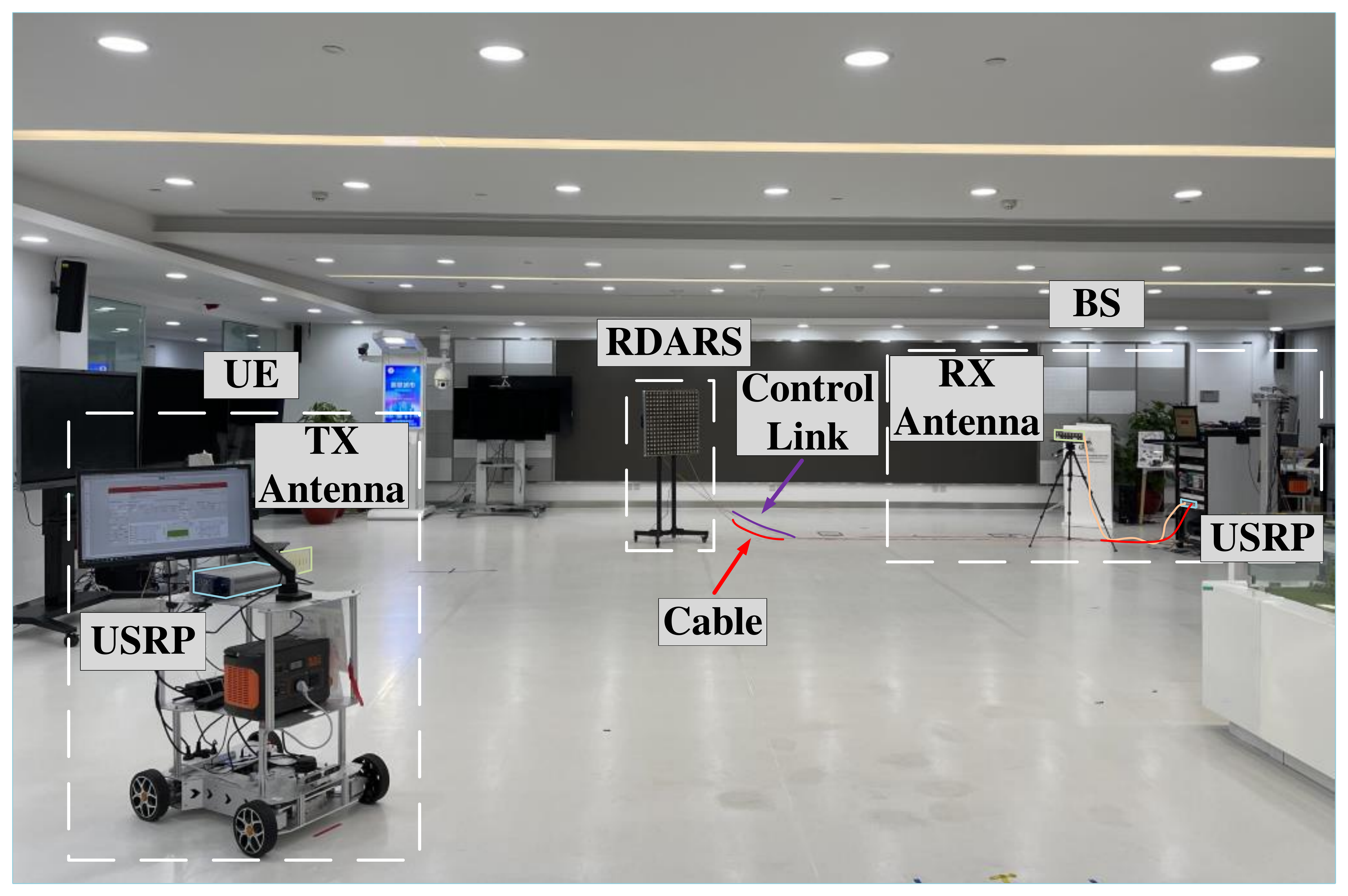} 
        }  
        \\
        \subfigure[RDARS-aided system] {
        \includegraphics[width=0.46\columnwidth]{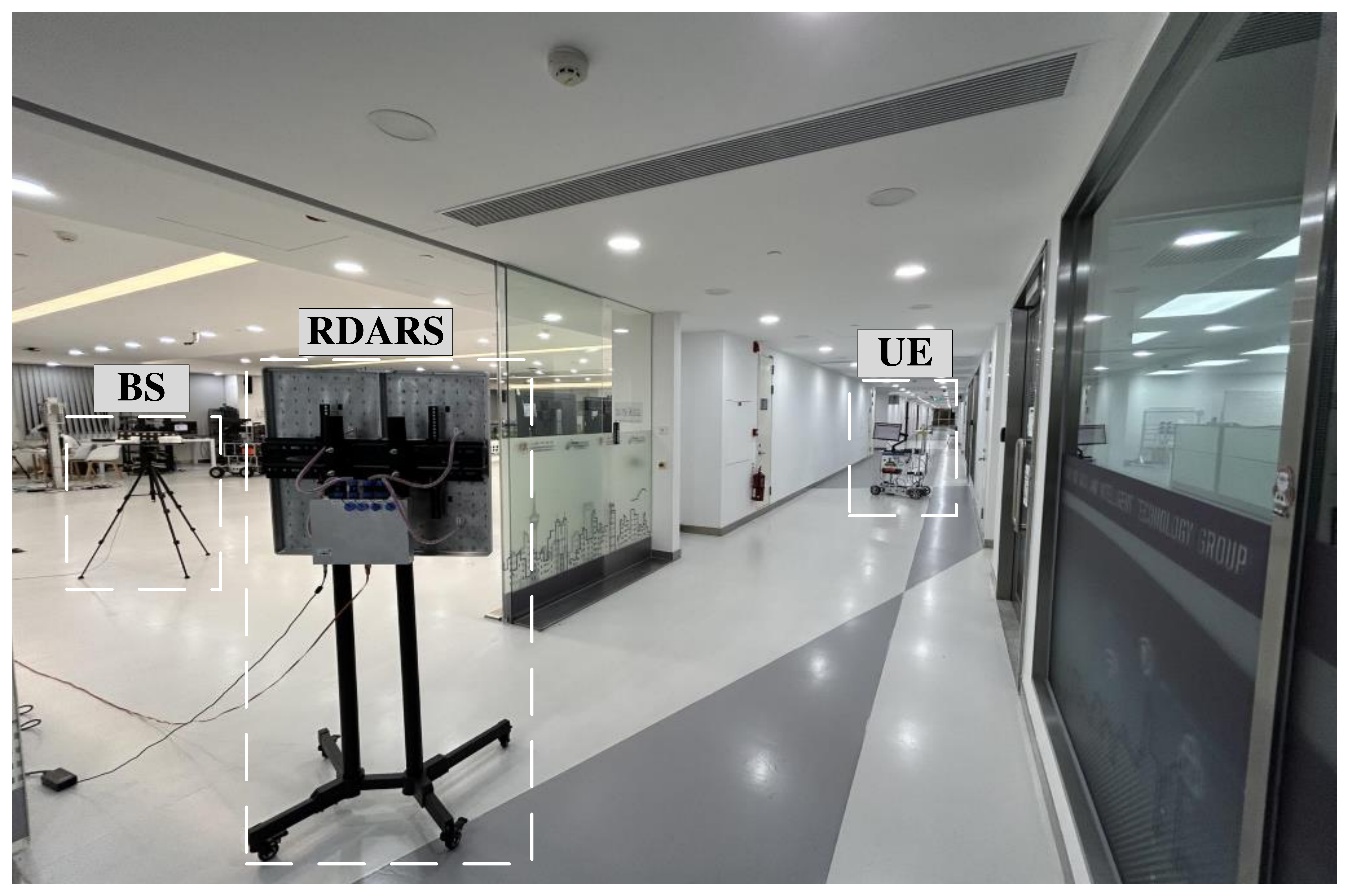} 
        }  
        \subfigure[DAS system] {
        \includegraphics[width=0.46\columnwidth]{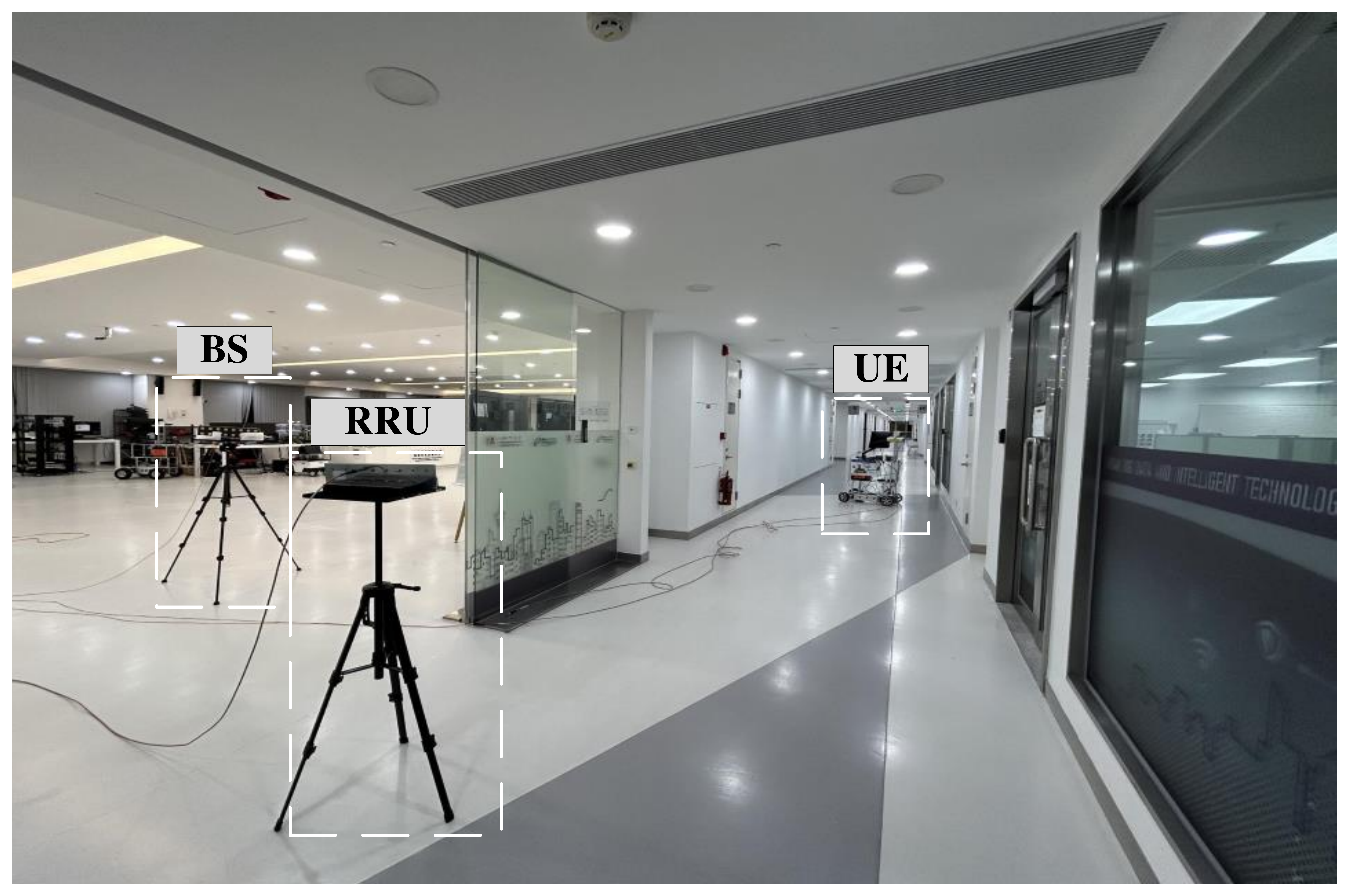} 
        } 
        \caption{ 
        Illustration of experiment scenarios. (a) \textbf{Scenario 1} with direct UE-BS link. (b), (c) \textbf{Scenario 2} without direct UE-BS link.}    
    \end{figure}

    \begin{figure} [htb]  
        \centering
        \setlength{\belowcaptionskip}{-3mm}
        \subfigure[Front view] {    
        \includegraphics[width=0.4\columnwidth]{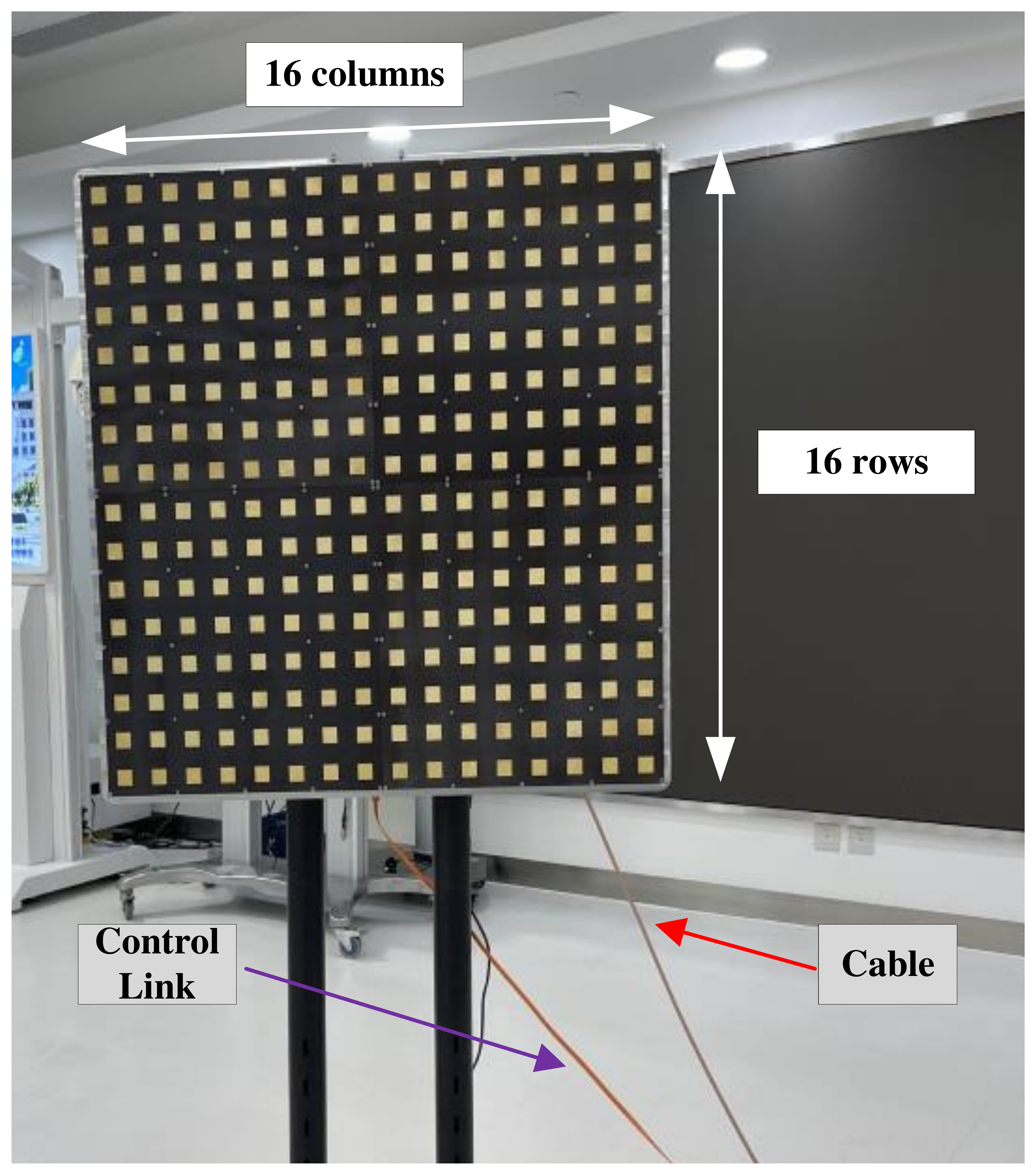} 
        }  
        \subfigure[Rear view] {
        \includegraphics[width=0.4\columnwidth]{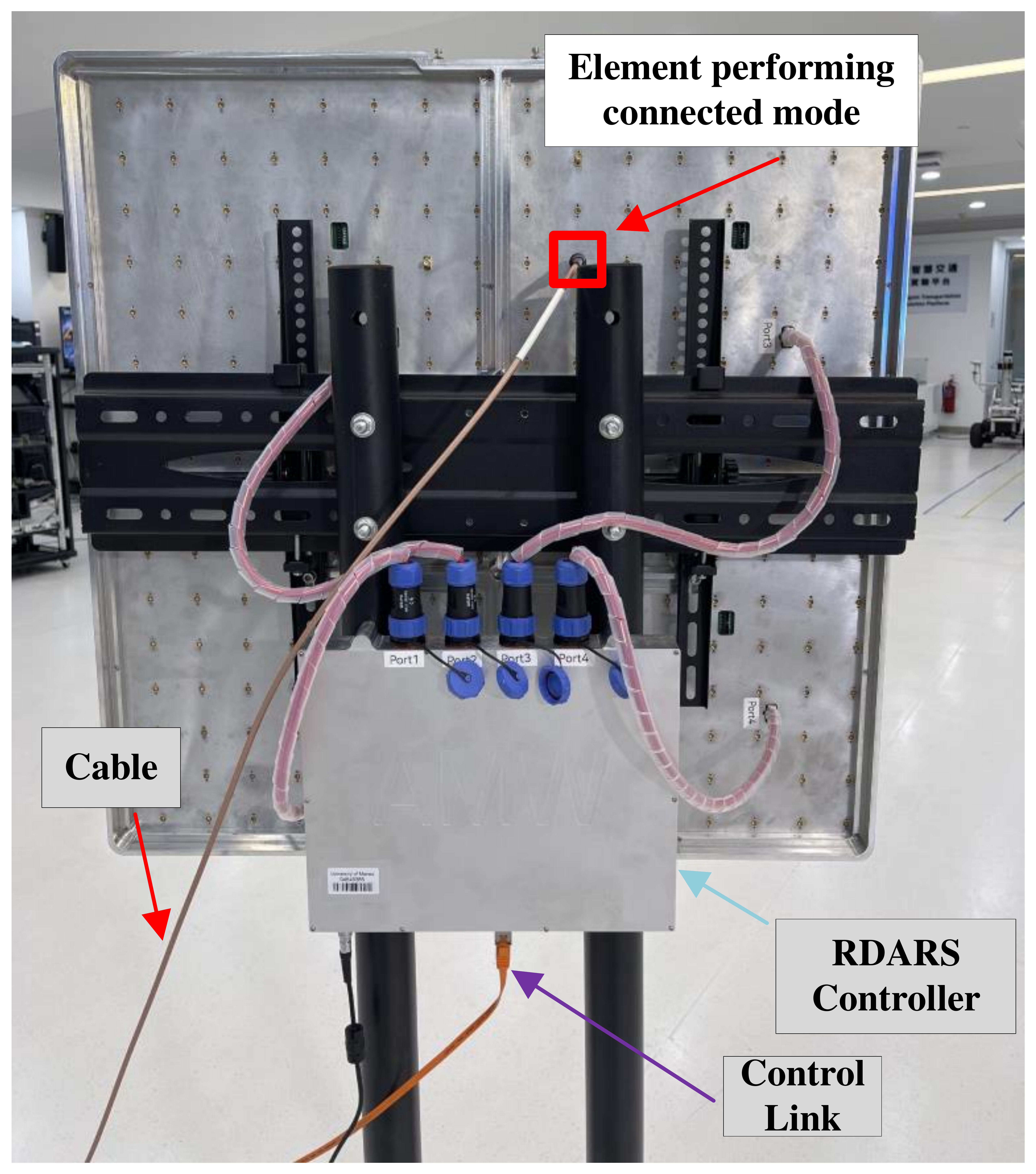} 
        }  
        \caption{ 
       Prototype of $16 \times 16$ RDARS, each element capable of 2-bit phase shifting under \textit{reflection mode}, and is connected to a cable when performing \textit{connected mode}. }    
    \end{figure}
    
    Specifically, the RDARS is made of 16 unit cells per column and per row, i.e., $16\times 16=256$ total elements. Each element of RDARS has $2$-bit phase resolution when performing \textit{reflection mode}, and will switch to the \textit{connected mode} when has been programmed, and cable is connected to its port in advance at the back of the surface as shown in Fig. 9(b). We apply patch antenna design for elements on RDARS similar to \cite{Xilong_Pei}. However, different from the existing realization of RIS where phase response is generated directly on the element with different voltage applied for a phase circuit, to realize the mode selection on RDARS, the patch antenna is first connected to a Complementary Metal-Oxide-Semiconductor (CMOS) switcher through near-lossless microstrip, where different states can be selected with the control circuit. The controlling signal for RDARS is delivered using User Datagram Protocol (UDP) with a total of $512$ bits.

    % \begin{figure} [htb]  
    %     \centering
    %     \setlength{\belowcaptionskip}{-3mm}
    %     \subfigure[RDARS-aided system, $a = 3$] {  
    %     \includegraphics[width=0.46\columnwidth]{./FigV1/RDARS_data.pdf} 
    %     }
    %     \subfigure[RIS-aided system, $a = 0$] {
    %     \includegraphics[width=0.46\columnwidth]{./FigV1/RIS_data.pdf} 
    %     }  
    %     \caption{ 
    %    Illustration of measurement results.}    
    % \end{figure}

    \begin{table} [htb]
        \caption{ System parameters and hardware modules}
        \centering
        \scalebox{0.8}{
        \begin{tabular}{ |c|c|  }
             \hline
             $\textbf{Parameter}$ &$\textbf{Value}$
              \\
             \hline
             Center frequency      & $3.7$ GHz   \\
             \hline
             Modulation scheme     & QPSK  \\
             \hline
             Bandwidth   & $20$ MHz  \\
             \hline
             Transmit power & \textbf{Scenario1}:$-8$ dBm \textbf{Scenario2}:$-11$ dBm \\
             \hline
             $\textbf{Name}$ &$\textbf{Description}$
              \\
             \hline
             Central controller &Computer (Intel Core i7-8700 processor) \\
             \hline
             SDR     & \makecell[c]{NI USRP-$2974$ (TX) \\
             NI USRP-$2954$ (RX)}   \\
             \hline
             RDARS   & \makecell[c]{Size: $70.5$cm $\times$ $70.5$cm ($256$ elements) \\ Operating frequency: $3.6-3.8$GHz} \\
             \hline
             Yagi-uda antenna (TX) & \makecell[c]{Gain $\approx$ $10$ dB $@3.7$ GHz \\ Beamwidth $\approx$ $40^{\circ}$} \\
             \hline
             Microstrip-patch antenna (RX) & \makecell[c]{Gain $\approx$ $6$ dB $@3.7 $ GHz \\ Beamwidth $\approx$ $90^{\circ}$} \\
             \hline
        \end{tabular}
        }
    \end{table} 

    For thorough comparison, two scenarios are tested, i.e., \textbf{Scenario 1} with UE-BS link, as shown in Fig. 8(a), and \textbf{Scenario 2} without UE-BS link (blockage), as shown in Fig. 8(b).
    For both scenarios, the performance of RDARS-aided, RIS-aided, and DAS systems are compared under the same system configurations.\footnote{ For RIS-aided system, the experiment is conducted by defaulting setting $a = 0$ on RDARS. For the DAS system, the RDARS is replaced by a RRU of a similar structure at the same location as shown in Fig. 8(c).} Since the RDARS prototype is designed to have a uniform planer array (UPA) form, the $\textit{optimal phase shifts}$ in both scenarios are determined by searching through the codebook and finding the codeword that gives the maximum received power, as mentioned in\cite{Xilong_Pei}. Once determined, the $\textit{optimal phase shifts}$ are fixed for each scenario during the experiment conducted in a short period for a fair comparison. 
    % where the LOS path dominates. 
    %As such, both incident and desired reflecting directions of the electromagnetic waves (EM waves) can be pre-determined in both vertical and horizontal directions by exploiting the geometric information. 
    % Since the RDARS is designed to have a Uniform Planer Array (UPA) form, the $\textit{optimal phase shift}$  is determined similarly to \cite[(12)]{Xilong_Pei} by first calculating the optimal continuous phase shifts and then applying quantization to $2$-bit resolution. 
    % first calculating the optimal continuous phase shifts and then applying quantization to $2$-bit resolution.
    % The $\textit{random phase shift}$ is also introduced by setting the optimal phase shift with an arbitrary offset.
    
    The experimental results are presented in Table ${\rm\uppercase\expandafter{\romannumeral2}}$ and ${\rm\uppercase\expandafter{\romannumeral3}}$ with three performance metrics, i.e., the received SNR\footnote{
    Here the received SNR refers to the symbol SNR in dB for the given sequence of IQ samples and is calculated based on Error-Vector-Magnitude (EVM), i.e.,  ${\rm{SNR}}$ [dB] $= -20 \log_{10}({\rm{EVM}})$.}, the data transmission throughput, and BLER. The throughput improvement in Table ${\rm\uppercase\expandafter{\romannumeral2}}$ and ${\rm\uppercase\expandafter{\romannumeral3}}$ is calculated between the RDARS-aided and its counterparts based on the $\textit{optimal phase shifts}$. 
    % One example of the throughput measurements in the experiments is shown in Fig. 10 for a clear illustration. 
    As seen in Table ${\rm\uppercase\expandafter{\romannumeral2}}$, RDARS is shown to improve system performance in \textbf{Scenario 1}, e.g., the RDARS with only one element performing \textit{connected mode} brought a $170\%$ throughput improvement, i.e., from $10.0$ Mbps to $27.0$ Mbps, as compared to the RIS-aided system. Also, the received SNR is improved by $2.2$ dB, i.e., from $10.6$ dB to $12.8$ dB, and the BLER is reduced by $63.0\%$, i.e., from $64\%$ to $1\%$. Moreover, a $10\%$ throughput improvement is attained by RDARS as compared to DAS with the same amount of distributed antenna, e.g., from $24.4$ Mbps to $27.0$ Mbps. While in \textbf{Scenario 2}, RDARS is shown to greatly improve system performance compared to its counterparts as seen in Table ${\rm\uppercase\expandafter{\romannumeral3}}$, e.g., the RDARS with only one element performing \textit{connected mode} brought a $560\%$ throughput improvement, i.e., from $3.5$ Mbps to $23.1$ Mbps, as compared to the RIS-aided system. Notably, a $21\%$ throughput improvement is attained by RDARS as compared to DAS with the same amount of distributed antenna, e.g., from $19.0$ Mbps to $23.1$ Mbps. The above results demonstrate the superiority of RDARS by simultaneously exploiting both \textit{distribution gain} and \textit{reflection gain} to combat the ``multiplicative fading'' encountered by the conventional RIS and enhance the performance of conventional DAS via low-cost passive components.

    % This demonstrates the superiority of RDARS by compensating the ``multiplicative fading" effect. Notably, the gain introduced by increasing the number of elements performing \textit{connected mode} shows a diminishing marginal effect, e.g., 4.44.4 Mbps and 11 Mbps throughput improvement from a=1a = 1 to a=2a = 2 and from a=2a = 2 to a=3a = 3, respectively. This indicates that for practical deployment of RDARS, only a few elements performing \textit{connected mode} are sufficient to greatly improve performance with moderate cost.

    The experimental results have validated the effectiveness of the proposed RDARS and its potential to assist wireless communications systems by significantly boosting the performance compared to its counterparts.
    \textit{In fact, in the conventional fully passive RIS setting, a control link is needed for delivering the control signal for phase shift design}. 
    % Though our prototype adopts two separate links for RDARS-controller (or RIS-controller by setting all the elements under \textit{reflection mode}) and the connection of elements performing \textit{connected mode} with BS (using dedicated cable), we envision that these two links could be integrated together.
    Research in \cite{Kangda_Zhi} has pointed out that the phase shift design based on statistical CSI can provide a similar gain compared to instantaneous CSI. Since the fronthaul used by \textit{connected mode} on RDARS and the control link of RDARS can be integrated with the DAS system. As such, the RDARS configuration (both $\boldsymbol{\theta}$ and $\mathbf{A}$) delivery through the control link only needs to be updated on a large-time scale where the remaining time block could be utilized to deliver the signal received at elements performing \textit{connected mode} on RDARS. This could potentially further boost the RDARS-aided system's performance in terms of achievable rate.

    \begin{table*} [htb]
    \caption{Experimental results in \textbf{Scenario 1}}
        \centering
        \begin{tabular}{ |c|c|c|c|c|c|  }
             \hline
             \multicolumn{2}{|c|}{\diagbox[innerwidth = 2cm]{Scenario}{Metric}}
             &Throughput[Mbps]
             &BLER 
             &Received SNR[dB]
             &\makecell[c]{Throughput \\ improvement} \\
             \cline{1-6}
             \multirow{1}*{RDARS} 
             &$a = 1$  &$27.0$ &$0.01$ &$12.8$  &\diagbox{}{}   \\
             \cline{1-6}
             \multirow{1}*{DAS} 
             &$a = 1$  &$24.4$ &$0.10$ &$12.5$  &$10\%$   \\
             \cline{1-6}
             \multirow{1}*{RIS} 
             &$a = 0$ &$10.0$ &$0.64$ &$10.6$ &$170\%$  \\
             \cline{1-6}
             {\makecell[c]{Without RIS/ RDARS \\ (removed)}}
             &\diagbox{}{} &$2.0$ &$0.92$ &$9.9$ &$1250\%$ \\
              \cline{1-6}
        \end{tabular}
    \end{table*} 

    \begin{table*} [htb]
    \caption{Experimental results in \textbf{Scenario 2}}
        \centering
        \begin{tabular}{ |c|c|c|c|c|c|  }
             \hline
             \multicolumn{2}{|c|}{\diagbox[innerwidth = 2cm]{Scenario}{Metric}}
             &Throughput[Mbps]
             &BLER 
             &Received SNR[dB]
             &\makecell[c]{Throughput \\ improvement} \\
             \cline{1-6}
             \multirow{1}*{RDARS} 
             &$a = 1$  &$23.1$ &$0.15$ &$11.4$  &\diagbox{}{}   \\
             \cline{1-6}
             \multirow{1}*{DAS} 
             &$a = 1$  &$19.0$ &$0.30$ &$10.9$  &$21\%$   \\
             \cline{1-6}
             \multirow{1}*{RIS} 
             &$a = 0$ &$3.5$ &$0.86$ &$9.7$ &$560\%$  \\
              \cline{1-6}
        \end{tabular}
    \end{table*} 

    \section{conclusion}

    This paper proposed a novel architecture, namely ``Reconfigurable distributed antennas and reflecting surfaces (RDARS)'', as a flexible and reconfigurable combination of distributed antennas and reflecting surfaces. Closed-form expressions for ergodic achievable rate have been derived with optimal and arbitrary RDARS configuration using MRC. Meaningful insights have been found and its superiority has been demonstrated by simulation results. A showcase using a prototype of the RDARS-aided system with 256 elements has been provided with experimental results. Both theoretical analysis and experimental results confirmed the effectiveness and the practicability of the proposed RDARS.

    \begin{appendices}
    
    \section{Proof of \textbf{Proposition 2}}
        Recall that $\gamma_s = \overline{\gamma}_s ( (\gamma_1 + \gamma_2)^2 + \gamma_3 )$, where we define $\gamma_1 = |h_{UB}|$, $\gamma_2 = \sum_{i=1}^{N}(1-a_i)|h_{RB,i}||h_{UR,i}|$, and $\gamma_3 = \mathbf{h}_{UR}^{H} \mathbf{A}^{H} \mathbf{A} \mathbf{h}_{UR}$. To obtain the first and second moments of $\gamma$, we need the following statistical information: the first to the fourth moments of $\gamma_1$, $\gamma_2$, and the first and the second moments of $\gamma_3$.
    To simplify the notation, we define $x_i=(1-a_i)|h_{UR,i}||h_{RB,i}|$. 
    Below we derive the aforementioned required statistics.

    For $\gamma_1$, the first to fourth moments of $h_{UB}$ can be easily obtained as $\mathbb{E}[\gamma_1] = \frac{1}{2}\sqrt{\pi} \gamma$, $\mathbb{E}[\gamma_1^2] =   \gamma^2$, $\mathbb{E}[\gamma_1^3] = \frac{3}{4} \sqrt{\pi} \gamma^{3}$ and $\mathbb{E}[\gamma_1^4] =  2\gamma^{4}$.

    For $\gamma_2$, the first to second moments are calculated as 
    $\mathbb{E}[\gamma_2]
            =
            \sum_{i=1}^{N}(1-a_i) \mathbb{E}[|h_{UR,i}|] \mathbb{E}[|h_{RB,i}|]=
            (N-a) \frac{\pi}{4} \alpha \beta$.
    $\mathbb{E}[\gamma_2^2] 
            =
            \mathbb{E}[\sum_{i=1}^{N} (1-a_i)^{2} |h_{UR,i}|^{2} |h_{RB,i}|^{2} \notag 
            + 
             \sum_{i \neq j}^{N} (1-a_i) (1-a_j) |h_{UR,i}| |h_{RB,i}| |h_{UR,j}| |h_{RB,j}|]
            =
            (N-a)(1+ \frac{\pi^2}{16}(N-a-1)) \alpha^{2} \beta^{2}$.
    $\mathbb{E}[\gamma_2^3] 
            =
            \mathbb{E}[\sum_{i=1}^{N} x_i^3 + \sum_{i\neq j}^N x_i^2 x_j + \sum_{i \neq j \neq k}^{N} x_i x_j x_k ]$.
    %%%
    For the third moment of $\gamma_2$, we have
    $ \mathbb{E}[\gamma_2^3] 
            =
            \mathbb{E}[\sum_{i=1}^{N} x_i^3 + \sum_{i\neq j}^N x_i^2 x_j + \sum_{i \neq j \neq k}^{N} x_i x_j x_k ]$,
    $\mathbb{E}[\sum_{i=1}^{N} x_i^3] = (N-a) \frac{9}{16} \pi \alpha^3 \beta^3$, $\mathbb{E}[\sum_{i\neq j}^N x_i^2 x_j] = {\rm{C}_{3}^{1}} {\rm{C}_{N-a}^{2}} {\rm{C}_{2}^{1}} \frac{1}{4} \pi \alpha^3 \beta^3$, and $\mathbb{E} [\sum_{i \neq j \neq k}^{N} x_i x_j x_k]  ={\rm{C}_{N-a}^{3}} {\rm{C}_{3}^{1}} {\rm{C}_{2}^{1}}
        \frac{1}{64} \pi^3 \alpha^3 \beta^3$ where ${\rm{C}_{N}^{M}} = \frac{N!}{M!(N-M)!}$. 
    Therefore we have $\mathbb{E} [\gamma_2^3] 
            =
            C_3 \alpha^3 \beta^3$, where $C_3 = (N-a) \frac{9}{16} \pi +
            3(N-a)(N-a-1) \frac{1}{4} \pi  +
            (N-a)(N-a-1)(N-a-2)  \frac{1}{64} \pi^3$.
    %%%
    For the fourth moment of $\gamma_2$, we have
    $\mathbb{E}[\gamma_2^4]
        = \mathbb{E}[
        \sum_{i=1}^{N} x_i^{4} 
        + \sum_{i \neq j}^{N} x_i^{3} x_j^{1} 
        + \sum_{i \neq j}^{N} x_i^{2} x_j^{2} 
        + \sum_{i \neq j \neq k }^{N} x_i^{2} x_j x_l
        + \sum_{i \neq j \neq k \neq l}^{N} x_i^{2} x_j^{2} x_k^{2} x_l^{2} ]$, 
    where $\mathbb{E} [ \sum_{i=1}^{N} x_i^{4}]  = 4(N-a) \alpha^{4} \beta^{4}$, $\mathbb{E} [ \sum_{i \neq j}^{N} x_i^{3} x_j^{1}] = {\rm{C}_{4}^{1}} {\rm{C}_{N-a}^{2}} {\rm{C}_{2}^{1}}        \frac{9}{64} \pi^{2} \alpha^{4} \beta^{4}$, $\mathbb{E} [ \sum_{i \neq j}^{N} x_i^{2} x_j^{2}] =  {\rm{C}_{4}^{2}} {\rm{C}_{N-a}^{2}}  \alpha^{4} \beta^{4}$, $\mathbb{E}[\sum_{i \neq j \neq k }^{N} x_i^{2} x_j x_l] = {\rm{C}_{N-a}^{3}} {\rm{C}_{4}^{2}} {\rm{C}_{3}^{1}} {\rm{C}_{2}^{1}} {\rm{C}_{2}^{1}} \frac{{\pi}^2}{16} \alpha^{4} \beta^{4}$, and $\mathbb{E} [ \sum_{i \neq j \neq k \neq l}^{N} x_i x_j x_k x_l] = {\rm{C}_{N-a}^{4}} {\rm{C}_{4}^{1}} {\rm{C}_{3}^{1}} {\rm{C}_{2}^{1}} 
        \frac{1}{256} \pi^{4} \alpha^{4} \beta^{4}$.
    Thus it yields $ \mathbb{E} [\gamma_2^4] = C_4 \alpha^{4} \beta^{4}$, where 
    $C_4 = 4(N-a) + \frac{9}{16}(N - a)(N-a-1) \pi^{2} + 3 (N-a) (N-a-1) + 6(N-a)(N-a-1)(N-a-2) \frac{{\pi}^2}{16} + \frac{1}{256}(N-a)(N-a-1)(N-a-2)(N-a-3) \pi^{4}$.
    %%%
    
    For the calculation of $\gamma_3$, we have $\mathbf{h}_{UI}^{H} \mathbf{A}^{H} \mathbf{A} \mathbf{h}_{UI} 
        =
        \sum_{i=1}^{N} a_i^2 |h_{UR,i}|^2 =
        \sum_{i=1}^{N} a_i^2 y_i$,
    where $a_i = \mathbf{A}(i,i)$ and $y_i = |h_{UR,i}|^2$. Note that $\frac{y_i}{ \frac{1}{2} \alpha^2 }$ follows Chi-square distribution with two degrees of freedom, i.e., $\frac{y_i}{ \frac{1}{2} \alpha^2 } \sim \chi^2(2)$. And $y_i = \frac{\alpha^2}{2} \frac{y_i}{ \frac{1}{2} \alpha^2 }$ follows Gamma distribution, i.e., $y_i \sim \Gamma(1, \alpha^2)$. Therefore we have $\mathbf{h}_{UR}^{H} \mathbf{A}^{H} \mathbf{A} \mathbf{h}_{UR} \sim \Gamma(a, \alpha^2)$, and the first and second moments are obtained as
    $\mathbb{E}[\gamma_3] = 
            \mathbb{E}[\mathbf{h}_{UR}^{H} \mathbf{A}^{H} \mathbf{A} \mathbf{h}_{UR}] = a \alpha^2$, 
    $\mathbb{E}[\gamma_3^2] = \mathbb{E}[(\mathbf{h}_{UR}^{H} \mathbf{A}^{H} \mathbf{A} \mathbf{h}_{UR})^2] =
            \sum_{i=1}^N a_i^2 \mathbb{E}[|h_{UR,i}|^4] + 
            \sum_{i \neq j}^N a_i a_j \mathbb{E}[|h_{UR,i}|^2] \mathbb{E}[|h_{UR,j}|^2] = a(a+1) \alpha^4$.
    With the above statistics, $\gamma$'s first and second moments are derived.

    \section{Proof of \textbf{Corollary 3}}

    The average received SNRs of the RDARS-aided system, RIS-aided system and DAS are obtained as 
    $\mathbb{E}[\gamma_{s}^{RDARS}] = \overline{\gamma}_s \big[ \gamma^2 +    (N-a) \frac{\pi}{4} \sqrt{\pi} \gamma \alpha \beta + (N-a)(1+ \frac{\pi^2}{16}(N-a-1)) \alpha^{2} \beta^{2} + a\alpha^2 \big]$, 
    $\mathbb{E}[\gamma_{s}^{RIS}] = \overline{\gamma}_s \big[ \gamma^2 + N\frac{\pi}{4} \sqrt{\pi} \gamma \alpha \beta 
            + (N)(1+ \frac{\pi^2}{16}(N-1)) \alpha^{2} \beta^{2}  \big]$,
    $\mathbb{E}[\gamma_{s}^{DAS}] =
            \overline{\gamma}_s \big[ \gamma^2 +  a\alpha^2 \big]$. 
    Let $\mathbb{E}[\gamma_{s}^{RDARS}]> \mathbb{E}[\gamma_{s}^{DAS}]$, we have $N > -\frac{4\sqrt{\pi}}{\pi}\frac{\gamma}{\alpha\beta} - \frac{16}{\pi^2} + a + 1$. It can be proved that $-\frac{4\sqrt{\pi}}{\pi}\frac{\gamma}{\alpha\beta} - \frac{16}{\pi^2} + a + 1 < a$. As such, we always have $\mathbb{E}[\gamma_{s}^{RDARS}]> \mathbb{E}[\gamma_{s}^{DAS}]$. Let $\mathbb{E}[\gamma_{s}^{RDARS}]> \mathbb{E}[\gamma_{s}^{RIS}]$, and further through some simple derivation we obtain (\ref{Lemma1_2}). Thus we complete the proof. 

    \section{Proof of \textbf{Theorem 4}}

    Here, we need to derive $E^{signal}$ and $E^{noise}$. We start by calculating $E^{signal} = \mathbb{E}[||\widetilde{\mathbf{h}}||^4]$ in the following. First, by identifying the non-zero terms, $\mathbb{E}[||\widetilde{\mathbf{h}}||^4]$ is decomposed as \eqref{E^{signal}_appendix} at the top of next page,  where we define $\underline{\mathbf{h}}_B \triangleq \mathbf{H}\mathbf{B}\mathbf{h} \in \mathbb{C}^{L \times 1}$. Then, the derivation of $E^{signal}$ follows by calculating these fourteen terms in \eqref{E^{signal}_appendix}. The first term $\widetilde{1}$ can be derived using a similar technique in \cite{Kangda_Zhi2} via the calculation of non-zero expectations among total $256$ terms, and therefore are omitted here for brevity. The remaining $\widetilde{2}$-$\widetilde{14}$ terms are calculated as $\widetilde{2} = \widetilde{3} = \gamma L 
            \big[ 
                    |f(\mathbf{A},\mathbf{\Theta})|^2 c\delta\epsilon 
                    +
                    (N-a) c\delta 
                    +
                    (N-a) c (\epsilon + 1 ) 
            \big]$,
    $\widetilde{4} = \gamma^2L(L+1)$,
    $\widetilde{5} = a^2 d^2 \epsilon^2$,
    $\widetilde{6} = a d^2 \epsilon$,
    $\widetilde{7} = a d^2 \epsilon$,
    $\widetilde{8} = a^2 d^2  + a d^2$,
    $\widetilde{9} = 
            2 L^2 |f(\mathbf{A},\mathbf{\Theta}) |^2 c \delta \epsilon \gamma 
            +
            2 L^2 (N-a) |f(\mathbf{A},\mathbf{\Theta}) |^2 c \delta \gamma 
            +
            2 L^2 (N-a) |f_{k}(\mathbf{A},\mathbf{\Theta}) |^2 c (\epsilon + 1) \gamma$,
    $\widetilde{10} = 
            2 L a |f(\mathbf{A},\mathbf{\Theta}) |^2 c d \delta \epsilon^2
            +
            2 L a (N-a) c d \delta \epsilon 
            +
            2 L a (N-a) c d \epsilon^2 
            +
            2 L a (N-a) c d \epsilon$,
    $\widetilde{11} = 
            2 L a |f(\mathbf{A},\mathbf{\Theta}) |^2 c d \delta \epsilon 
            +
            2 L a (N-a) c d \delta 
            +
            2 L a (N-a) c d \epsilon_k 
            +
            2 L a (N-a) c d$,
    $\widetilde{12} = 
            2 L a d \gamma \epsilon$,
    $\widetilde{13} = 
            2 L a d \gamma$,
    $\widetilde{14} = 
            2 a^2 d^2 \epsilon$. By rearranging the above derived results into compact form, we arrive at (\ref{E^{signal}}). 
    Next, $E^{noise} = \mathbb{E}[||\widetilde{\mathbf{h}}^{H}\mathbf{R}\widetilde{\mathbf{h}}||^2]$ is calculated as $\mathbb{E}[||\widetilde{\mathbf{h}}^{H}\mathbf{R}\widetilde{\mathbf{h}}||^2] = \mathbb{E} [ \sigma_B^2 (\underline{\mathbf{h}}_B + \mathbf{d} )^{H}
                        (\underline{\mathbf{h}}_B + \mathbf{d} )
                        +
                        \sigma_R^2
                        (\mathbf{A} \mathbf{h})^{H}
                        (\mathbf{A} \mathbf{h})]
    =
                        \sigma_B^2 \big[ 
                        \mathbb{E}[\underline{\mathbf{h}}_B^{H}\underline{\mathbf{h}}_B]
                        +
                        \mathbb{E}[\mathbf{d}^{H}\mathbf{d}]
                \big]
                +
                \sigma_R^2 \big[ 
                        \mathbb{E}[ \mathbf{h}^{H}\mathbf{A}^{H}
                        \mathbf{A} \mathbf{h}]
                \big]
    =
                \sigma_B^2 L \big[
                    |f(\mathbf{A},\mathbf{\Theta}) |^2 c \delta \epsilon
                    +
                    (N-a) c \delta 
                    +
                    \big( 
                        (N-a)c(\epsilon + 1) + \gamma
                    \big) 
                \big]
            +
               \sigma_R^2 a 
               d (\epsilon + 1)$. Thus we complete the proof.

    %   E^{signal}_appendix
    \newcounter{E^{signal}_appendix}
        \begin{figure*}[!htb]
        \small
        % ensure that we have normalsize text
        % \normalsize
        % Store the current equation number.
        \setcounter{E^{signal}_appendix}{\value{equation}}
        \setcounter{equation}{25}
        % \hrulefill
        \begin{align} \label{E^{signal}_appendix}
           \mathbb{E}[ ||\widetilde{\mathbf{h}}||^4]
           &=
                \underbrace{
                    \mathbb{E} \bigg[ 
                        \bigg| 
                            \underline{\mathbf{h}}_{B}^{H} \underline{\mathbf{h}}_{B}
                        \bigg|^2
                    \bigg]
                }_{\widetilde{1}}
            +
                \underbrace{
                    \mathbb{E} \bigg[ 
                        \bigg| 
                            \underline{\mathbf{h}}_{B}^{H} 
                            \mathbf{d}
                        \bigg|^2
                    \bigg]
                }_{\widetilde{2}}
            +   
                \underbrace{
                    \mathbb{E} \bigg[ 
                        \bigg| 
                            \mathbf{d}^{H}
                            \underline{\mathbf{h}}_{B} 
                        \bigg|^2
                    \bigg]
                }_{\widetilde{3}}
            +
                \underbrace{
                    \mathbb{E} \bigg[ 
                        \bigg| 
                            \mathbf{d}^{H}
                            \mathbf{d}
                        \bigg|^2
                    \bigg]
                }_{\widetilde{4}}
        \notag
            \\
            %%%%
        &+
                \underbrace{
                    d^2 \epsilon^2
                    \mathbb{E} \bigg[ 
                        \bigg| 
                            \overline{\mathbf{h}}^{H} 
                            \mathbf{A}^{H} \mathbf{A}
                            \overline{\mathbf{h}}
                        \bigg|^2
                    \bigg]
                }_{\widetilde{5}}
            +
                \underbrace{
                    d^2 \epsilon 
                    \mathbb{E} \bigg[ 
                        \bigg| 
                            \overline{\mathbf{h}}^{H} 
                            \mathbf{A}^{H} \mathbf{A}
                            \widetilde{\mathbf{h}}
                        \bigg|^2
                    \bigg]
                }_{\widetilde{6}}
            +
                \underbrace{
                    d^2 \epsilon 
                    \mathbb{E} \bigg[ 
                        \bigg| 
                            \widetilde{\mathbf{h}}^{H} 
                            \mathbf{A}^{H} \mathbf{A}
                            \overline{\mathbf{h}}
                        \bigg|^2
                    \bigg]
                }_{\widetilde{7}}
            +
                \underbrace{
                    d^2 
                    \mathbb{E} \bigg[ 
                        \bigg| 
                            \widetilde{\mathbf{h}}^{H} 
                            \mathbf{A}^{H} \mathbf{A}
                            \widetilde{\mathbf{h}}
                        \bigg|^2
                    \bigg]
                }_{\widetilde{8}}
        \notag
            \\
            %%%%
        &+  
                \underbrace{
                    2 {\rm{Re}} \bigg\{
                        \mathbb{E} \bigg[ 
                                \underline{\mathbf{h}}_{B}^{H} 
                                \underline{\mathbf{h}}_{B}
                                \mathbf{d}^{H}
                                \mathbf{d}
                        \bigg]
                    \bigg\}
                }_{\widetilde{9}}
            +
                \underbrace{
                    2 {\rm{Re}} \bigg\{ 
                        d \epsilon
                        \mathbb{E} \bigg[ 
                                \hat{\underline{\mathbf{h}}}_{B}^{H} 
                                \underline{\mathbf{h}}_{B}
                                \overline{\mathbf{h}}^{H}
                                \mathbf{A}^{H} \mathbf{A}
                                \overline{\mathbf{h}}
                        \bigg]
                    \bigg\}
                }_{\widetilde{10}}
            +
                \underbrace{
                    2 {\rm{Re}} \bigg\{ 
                        d 
                        \mathbb{E} \bigg[ 
                                \hat{\underline{\mathbf{h}}}_{B}^{H} 
                                \underline{\mathbf{h}}_{B}
                                \widetilde{\mathbf{h}}^{H}
                                \mathbf{A}^{H} \mathbf{A}
                                \widetilde{\mathbf{h}}
                        \bigg]
                    \bigg\}
                }_{\widetilde{11}}
        \notag
            \\
        &+
                \underbrace{
                    2 {\rm{Re}} \bigg\{ 
                        d \epsilon
                        \mathbb{E} \bigg[ 
                                \mathbf{d}^{H}
                                \mathbf{d}
                                \overline{\mathbf{h}}^{H}
                                \mathbf{A}^{H} \mathbf{A}
                                \overline{\mathbf{h}}
                        \bigg]
                    \bigg\}
                }_{\widetilde{12}}
            +
                \underbrace{
                    2 {\rm{Re}} \bigg\{ 
                        d 
                        \mathbb{E} \bigg[ 
                                \mathbf{d}^{H}
                                \mathbf{d}
                                \widetilde{\mathbf{h}}^{H}
                                \mathbf{A}^{H} \mathbf{A}
                                \widetilde{\mathbf{h}}
                        \bigg]
                    \bigg\}
                }_{\widetilde{13}}
            +   
                \underbrace{
                    2 {\rm{Re}} \bigg\{ 
                        d^2 \epsilon 
                        \mathbb{E} \bigg[ 
                                \overline{\mathbf{h}}^{H}
                                \mathbf{A}^{H}\mathbf{A}
                                \overline{\mathbf{h}}
                                \widetilde{\mathbf{h}}^{H}
                                \mathbf{A}^{H} \mathbf{A}
                                \widetilde{\mathbf{h}}
                        \bigg]
                    \bigg\}
                }_{\widetilde{14}}
        \end{align}
        % Restore the current equation number.
        \setcounter{equation}{\value{E^{signal}_appendix}}
        % The IEEE uses as a separator
        \hrulefill
        \end{figure*}
    \addtocounter{equation}{1}

    \end{appendices}

    \bibliographystyle{IEEEtran}
    \bibliography{main}

    \end{document}